\documentclass[11pt,letterpaper]{article}
\usepackage[margin=1in]{geometry}
\usepackage[dvipsnames,usenames]{xcolor}
\usepackage[colorlinks=true,pdfpagemode=UseNone,urlcolor=RoyalBlue,linkcolor=RoyalBlue,citecolor=OliveGreen,pdfstartview=FitH]{hyperref}

\definecolor{hkured}{HTML}{EE4123}
\definecolor{hkublue}{HTML}{009BD4}
\definecolor{hkugreen}{HTML}{00B38C}
\definecolor{hkuyellow}{HTML}{FED401}

\usepackage{todonotes}
\usepackage{tablefootnote}

\usepackage{amsfonts,amsmath,amsthm,thmtools,amssymb}
\usepackage{mathtools}

\declaretheorem[style=definition]{definition}

\declaretheorem{theorem}
\declaretheorem[sibling=theorem]{lemma}
\declaretheorem[sibling=theorem]{corollary}

\usepackage{tcolorbox}
\newtcolorbox{algorithm}[1]
{
	adjusted title = {#1},
	fonttitle = \bfseries,
	beforeafter skip = 12pt,
}

\usepackage{enumitem}
\usepackage{xspace}
\usepackage{booktabs}
\usepackage{nicefrac}
\usepackage{cleveref}
\usepackage{tabto}
\usepackage{mathdots}
\usepackage{multirow}
\usepackage{subcaption}
\usepackage{float}
\usepackage[round]{natbib}

\usepackage{tikz}
\usetikzlibrary{shapes.multipart}

\newcommand{\OPT}{\mathrm{OPT}}
\newcommand{\ALG}{\mathrm{ALG}}
\newcommand{\A}{\mathcal{A}}
\newcommand{\dif}[1]{\,\text{\rm d} #1}

\newcommand{\utility}{U}
\newcommand{\regularizer}{R}

\newcommand{\defeq}{\stackrel{\textrm{\tiny def}}{=}}

\title{The Long Arm of Nashian Allocation in Online $p$-Mean Welfare Maximization}
\author{
	Zhiyi Huang
	\footnote{The University of Hong Kong. Email: zhiyi@cs.hku.hk, cslee@cs.hku.hk}
	\and
	Chui Shan Lee
	\footnotemark[1]
	\and
	Xinkai Shu
	\footnote{Max Planck Institute for Informatics. Email: xshu@mpi-inf.mpg.de. This work was done when the author was at the University of Hong Kong.}
	\and
	Zhaozi Wang
	\footnote{New York University. Email: zw4460@nyu.edu. This work was done when the author was at
Shanghai Jiao Tong University.}
}
\date{April 2025}
\begin{document}

\begin{titlepage}
	\thispagestyle{empty}
	\maketitle
	
	\begin{abstract}
		\thispagestyle{empty}
		We study the online allocation of divisible items to $n$ agents with additive valuations for $p$-mean welfare maximization, a problem introduced by Barman, Khan, and Maiti~(2022).
Our algorithmic and hardness results characterize the optimal competitive ratios for the entire spectrum of $-\infty \le p \le 1$.
Surprisingly, our improved algorithms for all $p \le \nicefrac{1}{\log n}$ are simply the greedy algorithm for the Nash welfare, supplemented with two auxiliary components to ensure all agents have non-zero utilities and to help a small number of agents with low utilities.
In this sense, the long arm of Nashian allocation achieves near-optimal competitive ratios not only for Nash welfare but also all the way to egalitarian welfare.

	\end{abstract}
\end{titlepage}

\section{Introduction}

Resource allocation is a fundamental challenge in many domains of Computer Science, Economics, and beyond, such as cloud computing \citep[e.g.,][]{GhodsiZHKSS:NSDI:2011} and food banks \citep[e.g.,][]{AleksandrovAGW:IJCAI:2015}.
The primary objectives in these scenarios are efficiency and fairness.
For $-\infty \le p \le 1$, the $p$-means of the agents' utilities, known as the $p$-mean welfare, form an axiomatically justified family of objectives \citep[][Chapter 3]{Moulin:2004} with different tradeoffs between these two factors.
On one extreme when $p = 1$, this is the utilitarian welfare, i.e., the sum of the agents' utilities.
On the other extreme when $p = - \infty$, this is the egalitarian welfare, i.e., the minimum utility among the agents.
In the middle when $p = 0$ is the Nash welfare, which reconciles the two extremes and satisfies several notions of fairness including envy-freeness
and proportionality \citep{Varian:1974}.

This paper studies online resource allocation algorithms for maximizing the $p$-mean welfare among $n$ agents when the items arrive sequentially and need to be allocated immediately upon arrival.
This \emph{online $p$-mean welfare maximization} problem was proposed by \citet*{BarmanKM:AAAI:2022}, assuming that 1) the items are divisible, 2) the agents' utilities are additive, and 3) each agent's utility for receiving all items, a.k.a., its monopolist utility, is equal to $1$.
For the whole spectrum of $-\infty \le p \le 1$, \citet{BarmanKM:AAAI:2022} gave competitive online algorithms and hardness results. 
However, there are gaps between their algorithmic and hardness results.
For any negative constant $p$, the gap is polynomially large in the number of agents.

By contrast, the non-negative regime is better understood.
For any positive $p$, the problem is equivalent to a special case of the \emph{online matching with concave returns} problem introduced by \citet{DevanurJ:STOC:2012}.
Their online primal dual approach gave a $\nicefrac{1}{p}$ competitive algorithm even without the third assumption above about unit monopolist utilities.
We include an exposition of this algorithm in \Cref{app:positive-regime} to be self-contained, and also in \Cref{app:positive-nashian} an improvement for $0 < p \le \nicefrac{1}{\log n}$ using the third assumption. 
For Nash welfare ($p = 0$), \citet*{BanerjeeGGJ:SODA:2022} gave an $O(\log n)$-competitive algorithm,%
\footnote{%
	More precisely, their model assumes having predictions of the agents' monopolist utilities, which can be used to normalize the agents' utilities to restore the assumption of unit monopolist utilities.
	See also \citet*{HuangLSW:WINE:2023} for online algorithms that relaxed this assumption to allow the absence of predictions of agents' monopolist utilities but still require the max-min ratio of monopolist utilities to be bounded.
}
which is the best possible.

\subsection{Our Contributions}

\paragraph{Optimal Competitive Ratios (up to Lower-Order Terms).}
We improve the upper and lower bounds for the competitive ratios of online $p$-mean welfare maximization.
These bounds characterize the optimal competitive ratios for all $-\infty \le p \le 1$.
See \Cref{tab:summary} for a summary.

For $p \ge \nicefrac{1}{\log n}$, we prove that the $\nicefrac{1}{p}$ competitive ratio by \citet{DevanurJ:STOC:2012} is optimal up to a $\log \nicefrac{1}{p}$ term even if the instance satisfies the assumption of unit monopolist utilities.

For the Nashian regime, i.e., when $-\nicefrac{1}{\log n} \le p \le \nicefrac{1}{\log n}$, we improve the $O(\log^3 n)$ competitive ratio by \citet{BarmanKM:AAAI:2022} to $O(\log n)$, and give a corresponding hardness result matching it up to a doubly logarithmic factor. 
Note that for Nash welfare ($p = 0$), our result matches the optimal ratio of $O(\log n)$ obtained by \cite{BanerjeeGGJ:SODA:2022}.

Between the Nashian regime and the harmonic mean, i.e., when $-1 \le p \le - \nicefrac{1}{\log n}$, we give an online algorithm that is $n^{\nicefrac{|p|}{(|p|+1)}}$-competitive, omitting lower-order terms.
This is better than the state-of-the-art by \citet{BarmanKM:AAAI:2022} by a polynomial factor for $p \ge - \nicefrac{1}{4}$, and a poly-logarithmic factor for other values of $p$.
We complement the algorithmic result with a hard instance, showing that no online algorithm can do better up to a logarithmic term.
This hardness result is a polynomial improvement compared to the existing $n^{\nicefrac{|p|}{(2|p|+1)}}$ for any negative constant $p$.

Last but not least, from harmonic mean to egalitarian welfare, i.e., when $-\infty \le p \le -1$, we show that the optimal competitive ratio is $\sqrt{n}$ up to a lower-order term.
Our hardness result is a polynomial improvement on the existing $n^{\nicefrac{|p|}{(2|p|+1)}}$ bound by \citet{BarmanKM:AAAI:2022} for any finite $p$;
our algorithmic result is slightly better than theirs by a $\sqrt{\log n}$ factor.

\begin{table}[t]
\renewcommand{\arraystretch}{1.3}
\setlength{\tabcolsep}{0pt}
\centering
\caption{Summary of Results. We omit constant and lower-order terms for brevity, and defer the precise bounds to the subsequent sections.
The bounds on some rows were known for a special value of $p$: \citet{BanerjeeGGJ:SODA:2022} gave $\log n$ upper and lower bounds for $p = 0$, and \citet{BarmanKM:AAAI:2022} showed a $\sqrt{n}$ lower bound for $p = -\infty$.
}
\label{tab:summary}
\vspace{-6pt}
\begin{tabular}{r@{\hskip 3pt}c@{\hskip 3pt}l@{\hskip 20pt}c@{\hskip 12pt}l@{\hskip 20pt}c@{\hskip 12pt}l}
	\toprule
	&&& \multicolumn{2}{l}{Algorithmic Results} & \multicolumn{2}{l}{Hardness Results} \\
        \midrule
	$\nicefrac{1}{\log n} \le$ & $p$ & $\le 1$ & $\nicefrac{1}{p}$ & \citet{DevanurJ:STOC:2012} & $\nicefrac{1}{p}$ & \Cref{thm:hardness-positive} \\
	$-\nicefrac{1}{\log n} \le$ & $p$ & $\le \nicefrac{1}{\log n}$ & $\log n$ & \Cref{cor:nashian} & $\log n$ & \Cref{thm:hardness-positive} \\
	$-o(1) \le$ & $p$ & $\le - \nicefrac{1}{\log n}$\tablefootnote{The omitted terms are of lower order for $p \le -\omega(\nicefrac{\log\log n}{\log n})$. For $-O(\nicefrac{\log\log n}{\log n}) \le p \le -\nicefrac{1}{\log n}$, our results show that the optimal competitive ratio is poly-logarithmic, but do not characterize the degree of the poly-logarithms.} & $n^{|p|}$ & \Cref{thm:nashian} & $n^{|p|}$ & \Cref{thm:hardness-negative-almost-nashian} \\
	$-1 \le$ & $p$ & $\le -\Omega(1)$ & $n^{\frac{|p|}{|p|+1}}$ & \Cref{thm:nashian-to-harmonic} & $n^{\frac{|p|}{|p|+1}}$ & \Cref{thm:hardness-nashian-to-harmonic} \\
	$-\infty \le$ & $p$ & $\le -1$ & $\sqrt{n}$ & \Cref{thm:harmonic-to-egalitarian}; \citet{BarmanKM:AAAI:2022} & $\sqrt{n}$ & \Cref{thm:hardness-harmonic-to-egalitarian} \\
	\bottomrule
\end{tabular}
\end{table}

\vspace{-3pt}
\paragraph{Long Arm of Nashian Allocation.}
Besides the improved competitive ratios, we find it conceptually interesting that our algorithmic results for the whole spectrum of $-\infty \le p \le \nicefrac{1}{\log n}$ are obtained by just two algorithms, both based on the greedy algorithm for the Nash welfare.
\emph{A priori}, it is surprising that the greedy algorithm for $p = 0$ also achieves nearly optimal competitive ratios for other values of $p$, even those that are far from zero.
This can be viewed as further evidence for the effectiveness of Nash welfare maximization for balancing efficiency and fairness in the context of online optimization, echoing the unreasonable fairness from offline Nash welfare maximization showed by \citet{CaragiannisKMPSW:TEAC:2019}.

More precisely, the improvements for $-1 \le p \le \nicefrac{1}{\log n}$ are obtained by combining the greedy algorithm for the Nash welfare with the common idea of distributing a constant fraction of each item uniformly to all agents to ensure an $\Omega(\nicefrac{1}{n})$ base utility for the agents \citep[see, e.g.,][for previous analyses of this algorithm for $p = 0$]{BanerjeeGGJ:SODA:2022, HuangLSW:WINE:2023}.
The base utilities allow us to avoid the irregularity of $p$-mean welfare when an agent has zero utility, for any $p \le 0$. 

We further address the cases of $-\infty \le p \le -1$ by introducing an auxiliary component that considers the other extreme of the spectrum, greedily maximizing the egalitarian welfare with respect to the agents' \emph{regularized utilities}.
An agent's regularized utility is the sum of its utility and a regularization term that is linear in the agent's monopolist utility for the remaining items.
We stress that the greedy algorithm for the Nash welfare still plays the main role in this regime.
In particular, it ensures that at most $\tilde{O}(\sqrt{n})$ agents may have low regularized utilities; 
the auxiliary component is designed to help these agents by making regularized egalitarian allocation.

\vspace{-3pt}
\paragraph{Our Techniques.}
Next, we sketch the competitive analysis of the greedy algorithm for the Nash welfare.
Consider the algorithm's allocation and any alternative allocation.
Denote each agent $a$'s utilities for these allocations as $\utility_a$ and $\tilde{\utility}_a$ respectively.
We will prove that the ratio of these utilities, i.e., $\nicefrac{\tilde{\utility}_a}{\utility_a}$, is at most $O( \log n )$ averaging over all agents $a$.
This is easy to prove in hindsight, based on 1) the greedy criteria with respect to the Nash welfare, and the fact that 2) the agents' maximum and minimum utilities differ by at most an $O(n)$ multiplicative factor due to the $\Omega(\nicefrac{1}{n})$ base utilities from distributing a constant fraction of each item uniformly to all agents.
Nonetheless, this is powerful enough to derive all subsequent lemmas and competitive ratios in this paper; 
we call it the \emph{Fundamental Lemma of Nashian Allocation} (\Cref{lem:nashian}).

In particular, we will use this fundamental lemma to derive upper bounds for the number of agents whose utilities are smaller than a threshold (\Cref{lem:bad-agents-number}), and the number of agents whose regularized utilities are smaller than a threshold (\Cref{lem:critical-agent-number}).
Intuitively, these bounds are useful because the $p$-mean welfare for any negative $p<0$ may be seen as a softmin function over the agents' utilities.
Hence, lower bounding the $p$-mean welfare of the algorithm's allocation reduces to upper bounding the number of agents whose (regularized) utilities are too small.

Finally, our hardness results consider upper triangular instances that are prevalent in online resource allocation and online matching \citep*[e.g.,][]{KarpVV:STOC:1990, KalyanasundaramP:TCS:2000}.
We supplement the upper triangular instances with items arriving at the end to satisfy the assumption of unit monopolist utilities.

\subsection{Further Related Work}

Online resource allocation problems have been studied extensively, including online packing \citep{AlonAABN:TALG:2006, BuchbinderN:MOR:2009} and online matching \citep{KarpVV:STOC:1990, Mehta:FTTCS:2013, HuangTW:SIGecom:2024}.
\citet{Walsh:ADT:2011} and \citet{AleksandrovAGW:IJCAI:2015} started the study of online fair division with agent and item arrivals respectively.
See \citet{AleksandrovW:AAAI:2020} for a survey.

Closest to this paper is the work by \citet{BarmanKM:AAAI:2022}, which we have already discussed and compared against.
Better results have been achieved in relaxed models of online allocation.
\citet{CohenP:ICALP:2023} studied the problem in a learning-augmented setting, i.e., where the algorithm has access to some extra (machine-learned) information. 
\citet{HajiaghayiPKS:NeurIPS:2022} considered the case of egalitarian welfare and indivisible items, assuming that the items arrive by a random order and the instance satisfies a large-market assumption.

Another line of work considered online resource allocation
for other notions of fairness.
\citet{BenadeKPP:EC:2018} achieved envy-freeness in online allocation in the sense of no-regret learning.
\citet{ZengP:EC:2020} explored the fairness-efficiency tradeoffs in online fair division with respect to envy-freeness and Pareto optimality.
\citet{GkatzelisPT:AAAI:2021} considered the online allocation of divisible items (to two agents) to satisfy envy-freeness and at the same time approximately maximize utilitarian welfare.
\citet{BanerjeeGHJMS:IJCAI:2023} studied online allocation of public goods for proportional fairness, with predictions for the agents' monopolist utilities.

Finally, offline maximization for $p$-mean welfare has also been studied in a long line of work.
Since the case of additive valuations and divisible items reduces to standard convex optimization, researchers focused on the harder case with indivisible items and non-additive valuations.
For general $-\infty \le p \le 1$, \citet{BarmanBKS:ESA:2020} and \citet{ChaudhuryGM:AAAI:2021} independently gave $O(n)$-approximation algorithms for maximizing the $p$-mean welfare when the valuations are subadditive.

For $p = 0$, maximizing the Nash welfare is APX-hard \citep{Lee:IPL:2017}.
The state-of-the-art for additive valuations is an $e^{\nicefrac{1}{e}}$-approximation algorithm by \citet{BarmanKV:EC:2018}, which has been generalized to the weighted case by \citet{FengL:ICALP:2024}.
\citet{LiV:FOCS:2022} obtained a constant approximation for submodular valuations while \citet{GargKK:SODA:2020} showed an $\nicefrac{e}{(e-1)}$ lower bound.
Recently, \citet{DobzinskiLRV:STOC:2024} gave a constant approximation for subadditive valuations.

For $p = -\infty$, maximizing the egalitarian welfare for additive valuations and indivisible items, a.k.a., the Santa Claus problem \citep[e.g.,][]{BansalS:STOC:2006}, is also APX-hard \citep{BezakovaD:SIGecom:2005}.
\citet{ChakrabartyCK:FOCS:2009} gave an $\tilde{O}(n^\varepsilon)$-approximate algorithm that runs in $n^{O(\nicefrac{1}{\varepsilon})}$ time, for any $\varepsilon=\Omega(\nicefrac{\log\log n}{\log n})$.

\section{Preliminaries}
\label{sec:prelim}

We write $\log$ for the natural logarithm in this paper.

\subsection{Model}

\paragraph{Discrete Time.}
Consider the problem of online allocation of $m$ \emph{divisible} items $I$ to $n$ agents $A$.
Agent $a$ has value $v_{ai} \ge 0$ for item $i$.
The items arrive one by one.
The agents' values for an item are unknown initially and revealed when the item arrives.
Upon an item's arrival, the algorithm must allocate it to some agent(s) immediately and irrevocably.
Denote the algorithm's allocation as $x = ( x_{ai} )_{a \in A, i \in I}$, where $x_{ai} \ge 0$ is the portion of item $i$ allocated to agent $a$.
Agent $a$'s \emph{utility} for allocation $x$ is:
\[
	\utility_a = \sum_{i \in I} v_{ai} \cdot x_{ai} 
	~.
\]

For some $p \le 1$, we want to maximize the \emph{$p$-mean welfare}:
\begin{equation}
	\label{eqn:p-norm-welfare}
	\left( \frac{1}{n} \sum_{a \in A} \utility_a^p \right)^{\frac{1}{p}}
	~.
\end{equation}

When $p = 1$, this is the \emph{utilitarian welfare} (up to a $\nicefrac{1}{n}$ factor).
When $p = 0$, \Cref{eqn:p-norm-welfare} is defined as the geometric mean of the agents' utilities, and is known as the \emph{Nash welfare}.
When $p = -1$, this becomes the harmonic mean of the agents' utilities;
we therefore call it the \emph{harmonic welfare}.
Last but not least, when $p = -\infty$, \Cref{eqn:p-norm-welfare} is the minimum utility among the agents, and is called the \emph{egalitarian welfare}.

\paragraph{Continuous Time.}
It will be more convenient to present our algorithms and analyses in the more general continuous-time model.
Next, we present the model and then explain the reduction from the discrete-time model to the continuous-time model.
Consider a continuum of infinitesimal items arriving in time horizon $[0, T)$ with unit arrival rate.
We will refer to the item arriving at time $t$ as item $t$, and thus, let $I = [0, T)$ denote the set of items.
Agent $a$ has unit value $v_a(t)$ for item $t$; we assume that $v_a(t)$ is piecewise constant, changing its value only a finite number of times.
Denote the algorithm's allocation as $x = \big( x_a(t) \big)_{a \in A, t \in I}$ where $x_a(t)$ is the portion of item $t$ allocated to agent $a$.
Correspondingly, agent $a$'s utility for allocation $x$ is:
\[
	\utility_a = \int_0^T v_a(t) x_a(t) \dif{t}
	~.
\]

Given any instance in the discrete-time model where we without loss of generality denote the set of $m$ divisible items as $\{1, 2, \dots, m\}$, we can reinterpret it in the continuous-time model with $T = m$, such that agent $a$'s value for items $i-1 \le t < i$ is $v_a(t) = v_{ai}$, for any $1 \le i \le m$.

\paragraph{Assumption of Unit Monopolist Utilities.}
Following \cite{BarmanKM:AAAI:2022}, we assume that the agents have unit monopolist utilities.
That is, for any agent $a \in A$, we have:
\[
	\int_0^T v_a(t) \dif{t} = 1
	~.
\]

Trivial hardness results exist if the algorithm has no information on the agents' monopolist utilities.
\citet{BanerjeeGGJ:SODA:2022} provided a hard instance that prevents any algorithm from performing better than $\Omega(n)$-competitive in online Nash welfare maximization ($p = 0$). The proof of $\Omega(n)$ hardness can be generalized to any $p \leq 0$ using the same hard instance.

Our results still hold under a weaker assumption that the monopolist utilities are known and within a constant factor of each other (see \Cref{sec:discussion}). This corresponds to having predictions for the agents' monopolist utilities, which could be derived from historical data and machine learning models. Moreover, the agents are of similar importance in the market.
Further relaxing this assumption is an interesting research direction. See \citet{BanerjeeGGJ:SODA:2022} and \citet{HuangLSW:WINE:2023} for some related results for the Nash welfare.

\subsection{Competitive Analysis}

\paragraph{Offline Optimal Allocation.}
We will compare the $p$-mean welfare of the algorithm's allocation to the offline optimal benchmark, obtained from optimizing the allocation based on full knowledge of the instance.
We can compute the offline optimal benchmark by solving a convex program:
\begin{align*}
	\text{maximize} \quad & \bigg( \frac{1}{n} \sum_{a \in A} \utility_a^p \bigg)^{\frac{1}{p}} \\[1ex]
	\text{subject to} \quad & \sum_{a \in A} x_a(t) \le 1 && \forall t \in I \\
	& \utility_a = \int_0^T v_a(t) x_a(t) \dif{t} && \forall a \in A \\[2ex]
	& x_a(t) \ge 0 && \forall a \in A, \forall t \in I	\notag
\end{align*}

We denote the optimal allocation as $x^* = \big( x_a^*(t) \big)_{a \in A, t \in I}$ and the corresponding agents' utilities and $p$-mean welfare as $\{U_a^*\}_{a\in A}$ and $\OPT$ respectively.
If there are multiple allocations that achieve the optimal $p$-mean welfare, we will pick an arbitrary one as $x^*$.

\paragraph{Competitive Online Algorithms.}
Denote the $p$-mean welfare for the online algorithm's allocation as $\ALG$.
An online algorithm is $\Gamma$-competitive if for every instance of online $p$-mean welfare maximization, we have:
\[
	\OPT \le \Gamma \cdot \ALG
	~.
\]

\subsection[Relaxation of Online p-Mean Welfare Maximization]{Relaxation of Online $p$-Mean Welfare Maximization}

We can guarantee a simple lower bound for all agents' utilities with the Uniform Allocation below.

\begin{algorithm}{Uniform Allocation}
	For each item $t \in I$, allocate it uniformly to all agents, i.e., $x_a(t) = \nicefrac{1}{n}$.
\end{algorithm}

\begin{lemma}
    \label{lem:unif}
	Every agent $a \in A$ gets utility $\nicefrac{1}{n}$ from 	Uniform Allocation.
\end{lemma}

As a result of this simple bound, we can consider a relaxation of online $p$-mean welfare maximization, in which each agent starts with $\nicefrac{1}{n}$ base utility for the algorithm's allocation.
In other words, an agent $a$'s utility for an allocation $x$ is:
\[
	U_a \defeq \frac{1}{n} + \int_0^T v_a(t) x_a(t) \dif{t}
	~.
\]

We further write $\utility_a(t)$ for agent $a$'s utility for the items it received before time $t$, i.e.:
\[
	\utility_a(t) \defeq \frac{1}{n} + \int_0^t v_a(s) x_a(s) \dif{s}
	~.
\]

As for the benchmark, we will still compare against the original optimal $p$-mean welfare without the $\nicefrac{1}{n}$ base utility.

\begin{lemma}
	\label{lem:relaxation}
	If an online algorithm $A$ is $\Gamma$-competitive for the relaxed online $p$-mean welfare maximization problem, then allocating half of each item by Uniform Allocation and the other half by algorithm $A$ is $2 \Gamma$-competitive for the original problem.
\end{lemma}

\section[Nashian Regime]{Nashian Regime: $- o(1) \le p \le \nicefrac{1}{\log n}$}
\label{sec:nashian}

This section considers the Nashian regime where  $p$ is close to zero. 
We will analyze the greedy algorithm for maximizing the Nash welfare and prove that it is competitive for all values of $p$ in this regime simultaneously. 

Recall that the Nash welfare is the case of $p = 0$, defined as the geometric mean of the agents' utilities.
Equivalently, we may consider maximizing its logarithm (up to a factor $n$):
\[
        \sum_{a \in A} \log \utility_a
	~.
\]

For each item $t$, allocating it to agent $a$ would increase this objective by:
\[
	\frac{v_a(t)}{\utility_a(t)}
	~.
\]

Hence, the Nashian Greedy algorithm allocates each item to an agent to maximize the above increment.
We present below a formal definition of the algorithm.

\begin{algorithm}{Nashian Greedy}
	\emph{Initialization:~}
	Let $x_a(t) = 0$, $\utility_a(0) = \frac{1}{n}$ for any agent $a \in A$ and any time $t \in I$.\\[2ex]
	\emph{Online Decisions:~}
	Allocate each item $t \in I$ to an agent $a$ that maximizes:
	\[
		\frac{v_a(t)}{\utility_a(t)}
		~,
	\]
	to greedily maximize the (logarithm of) Nash welfare:
	\[
		\sum_{a \in A} \log \utility_a(t)
		~.
	\]
\end{algorithm}

The main results of this section are the following theorem and its corollary when $|p| \le \nicefrac{1}{\log n}$.

\begin{theorem}
	\label{thm:nashian}
	For any $|p| = o(1)$, Nashian Greedy is $O \big( n^{|p|} \cdot \log n \big)$-competitive for the relaxed online $p$-mean welfare maximization problem.
\end{theorem}

\begin{corollary}
	\label{cor:nashian}	
	For any $p$ such that $|p| \le \nicefrac{1}{\log n}$, Nashian Greedy is $O(\log n)$-competitive for the relaxed online $p$-mean welfare maximization problem.
\end{corollary}

Next, we present the most important lemma for the analysis of Nashian Greedy.
Despite its simple form and proof, it is the foundation of all subsequent lemmas and competitive analyses in this paper.

\begin{lemma}[Fundamental Lemma of Nashian Allocation]
	\label{lem:nashian}
	Consider any allocation of the items $\tilde{x} = (\tilde{x}_{ai})_{a \in A, i \in I}$.
	For any time $t \in I$ and the agents' utilities for allocation $\tilde{x}$ up to time $t$, denoted as $\tilde{\utility}(t) = (\tilde{\utility}_a(t))_{a \in A}$, we have:
	\[
		\frac{1}{n} \sum_{a \in A} \frac{\tilde{\utility}_a(t)}{\utility_a(t)} \le \log (n+1)
		~.
	\]
\end{lemma}

\begin{proof}
	Consider any item $s \in [0, t)$.
	Suppose that Nashian Greedy allocates it to agent $a$.
	By the definition of the algorithm, for any agent $\tilde{a}$ we have:
	\[
		\frac{v_a(s)}{\utility_a(s)} \ge \frac{v_{\tilde{a}}(s)}{\utility_{\tilde{a}}(s)}
		~.
	\]
	
	The left-hand-side equals the increment of the logarithm of Nash  welfare for Nashian Greedy's allocation.
	Further, we relax the denominator of the right-hand-side to $U_{\tilde{a}}(t)$.
	We get that:
	\[
		\frac{\dif}{\dif{s}} \sum_{a \in A} \log \utility_a(s) \ge \frac{v_{\tilde{a}}(s)}{\utility_{\tilde{a}}(t)}
		~.
	\]
	
	Multiplying this inequality by $\tilde{x}_a(s)$ and summing over all agents $\tilde{a} \in A$, we get that:
	\[
		\frac{\dif}{\dif{s}} \sum_{a \in A} \log \utility_a(s)  \ge \frac{\dif}{\dif{s}} \sum_{a\in A} \frac{\tilde{\utility}_a(s)}{\utility_a(t)}
	\]
	
	Integrating over $s \in [0, t)$ gives:
	\[
		\sum_{a \in A} \log \frac{\utility_a(t)}{\utility_a(0)} \ge \sum_{a \in A} \frac{\tilde{\utility}_a(t)}{\utility_a(t)} 
		~.
	\]
	
	The lemma now follows by $1+\nicefrac{1}{n} \ge \utility_a(t) \ge \utility_a(0) = \nicefrac{1}{n}$.
\end{proof}

\begin{proof}[Proof of \Cref{thm:nashian}]
	For $p = 0$, the theorem follows by considering \Cref{lem:nashian} with $\tilde{x} = x^*$, i.e., the optimal allocation, and applying the AM-GM inequality to the left-hand-side.
	
	Next, we consider the case when $p \ne 0$.
	Define an auxiliary allocation $\tilde{x}$ that distributes half of each item uniformly to all agents, and the other half following the optimal allocation $x^*$.
	That is, for any agent $a \in A$ and any item $t \in I$:
	\[
		\tilde{x}_a(t) \defeq \frac{1}{2n} + \frac{1}{2} x^*_a(t)
		~.
	\]
	
	Let $\tilde{U}_a$ be agent $a$'s utility for allocation $\tilde{x}$.
	By definition, we have:
	\begin{equation}
		\label{eqn:auxiliary-utility-range}
		\frac{1}{2n} \le \tilde{U}_a \le \frac{1}{2} + \frac{1}{2n}	
		~,
	\end{equation}
	and also:
	\begin{equation}
		\label{eqn:auxiliary-allocation-approximation}
		\bigg( \frac{1}{n} \sum_{a \in A} \tilde{\utility}_a^p \bigg)^{\frac{1}{p}} \ge \frac{\OPT}{2}
		~.
	\end{equation}
	
	We write the $p$-mean welfare of Nashian Greedy's allocation as:
	\begin{equation}
		\label{eq:nashian-p-mean-rewrite}	
		\ALG = \bigg( \frac{1}{n} \sum_{a \in A} \utility_a^p \bigg)^{\frac{1}{p}} = \bigg(\frac{1}{n} \sum_{a \in A} \tilde{\utility}_a^p \cdot \Big( \frac{\tilde{\utility}_a}{\utility_a} \Big)^{-p} \bigg)^{\frac{1}{p}}
		~.
	\end{equation}
	
	Next, we introduce a set of auxiliary variables to denote:
	\[
		z_a \defeq \frac{\tilde{\utility}_a^p}{\sum_{a' \in A} \tilde{\utility}_{a'}^p}
		~.
	\]
	
	Comparing \Cref{eqn:auxiliary-allocation-approximation,eq:nashian-p-mean-rewrite}, it suffices to show that:
	\[
		\bigg( \sum_{a \in A} z_a \cdot \Big( \frac{\tilde{\utility}_a}{\utility_a} \Big)^{-p} \bigg)^{-\frac{1}{p}} \le (n+1)^{|p|} \cdot \log(n+1)
		~.
	\]
	
	By the definition of these auxiliary variables, we have $\sum_{a \in A} z_a = 1$.
	Hence, by applying the generalized mean inequality relating $(-p)$-mean and $1$-mean, where recall that $-p < 1$, we have:
	\[
		\bigg( \sum_{a \in A} z_a \cdot \Big( \frac{\tilde{\utility}_a}{\utility_a} \Big)^{-p} \bigg)^{-\frac{1}{p}}\le \sum_{a \in A}  z_a \cdot \frac{\tilde{\utility}_a}{\utility_a}
		~.
	\]
	
	Further, by the range of $\tilde{\utility}_a$ in \Cref{eqn:auxiliary-utility-range}, we have:
	\[
		z_a \le \frac{1}{n} \bigg( \frac{\max_{a \in A} \tilde{\utility}_a}{\min_{a \in A} \tilde{\utility}_a} \bigg)^{|p|} \le \frac{1}{n} \cdot (n+1)^{|p|}
		~.
	\]
	
	Combining the above with \Cref{lem:nashian}, we get that:
	\[
		\bigg( \sum_{a \in A} z_a \cdot \Big( \frac{\tilde{\utility}_a}{\utility_a} \Big)^{-p} \bigg)^{-\frac{1}{p}}\le (n+1)^{|p|} \cdot \frac{1}{n} \sum_{a \in A} \frac{\tilde{\utility}_a}{\utility_a} \le (n+1)^{|p|} \cdot \log(n+1)
		~.
	\]
\end{proof}

\section[From Nash to Harmonic Welfare]{From Nash to Harmonic Welfare: $-1 \le p \le - \Omega(1)$}
\label{sec:nash-to-harmonic}

In this section, we  continue to analyze the Nashian Greedy algorithm.
Surprisingly, it remains competitive even when $p$ is bounded away from zero.
The resulting competitive ratios are nearly optimal for all values of $p$ from Nash welfare to harmonic welfare.

\begin{theorem}
	\label{thm:nashian-to-harmonic}
	For any $-1 \le p \le -	\Omega(1)$, Nashian Greedy is:
	\[
		O \Big( n^{\frac{|p|}{|p|+1}} (\log n)^{\frac{1}{|p|+1}} \Big)
	\]
	competitive for the relaxed online $p$-mean welfare maximization problem.
\end{theorem}

\subsection{Further Properties of Nashian Allocation}

\begin{definition}[Bad Agents]
\label{def:bad}
	For any time $t\in I$ and any $\beta > 0$, let $B_\beta(t)$ be the set of \emph{$\beta$-bad} agents whose utilities at time $t$ are at most a $\beta$ fraction of the optimal $p$-mean welfare, i.e.:
	\[
	    B_{\beta}(t) \defeq \Big\{a\in A:\,\utility_a(t) \leq \beta \cdot \OPT \Big\}
	    ~.
	\]
	Further, we write $B_\beta$ for $B_\beta(T)$, the set of $\beta$-bad agents at the end.
\end{definition}

The next lemma considers a subset of $\beta$-bad agents, and upper bounds the sum of their utilities for the optimal allocation by a linear combination of $\beta \cdot \OPT$, the upper bound of these agents' utilities for the algorithm's allocation, and the sum of these agents' remaining monopolist utilities.

\begin{lemma}
	\label{lem:bad-agents-optimal-utility}
	For any time $t\in I$, any $\beta > 0$, and any subset of $\beta$-bad agents $S \subseteq B_\beta(t)$, we have:
	\[
		\frac{1}{n} \sum_{a\in S} \utility_a^* 
		\leq \beta \log (n+1) \cdot \OPT + \frac{1}{n} \sum_{a\in S} \Big(1-\int_0^t v_a(s) \dif{s} \Big)
	~.
	\]
\end{lemma}

\begin{proof}
	Recall that $x^*$ denotes the optimal allocation.
	The left-hand-side of the lemma's inequality can be written as:
	\[
		\frac{1}{n} \sum_{a \in S} \int_0^T v_a(s) x^*_a(s) \dif{s}
		~.
	\]
	
	On one hand, observe that:
	\[
		\int_t^T v_a(s) x^*_a(s) \dif{s} \le \int_t^T v_a(s) \dif{s} = 1 - \int_0^t v_a(s) \dif{s}
		~.
	\]
	
	On the other hand, by \Cref{lem:nashian}, we have:
	\[
		\frac{1}{n} \sum_{a \in A} \frac{1}{\utility_a(t)} \int_0^t v_a(s) x^*_a(s) \dif{s} \le \log(n+1) 
		~.
	\]
	
	We now drop all agents outside subset $S$ from the summation on the left-hand-side, and relax $\utility_a(t)$ to its upper bound $\beta \cdot \OPT$ for the remaining $\beta$-bad agents $a \in S$.
	We get that:
	\[
		\frac{1}{n} \sum_{a \in S} \int_0^t v_a(s) x^*_a(s) \dif{s} \le \beta \log (n+1) \cdot \OPT
		~.
	\]	
	
	Putting together these two parts proves the lemma.
\end{proof}

\begin{lemma}
	\label{lem:bad-agents-number}	
	For any $\beta > 0$, the fraction of $\beta$-bad agents at the end is at most:
	\[
		\frac{|B_\beta|}{n} \le \big( \beta \log (n+1) \big)^{\frac{|p|}{|p|+1}}
		~.
	\]
\end{lemma}

\begin{proof}
	By \Cref{lem:bad-agents-optimal-utility} with $t = T$ and $S = B_\beta$, we have:
	\[
		\frac{1}{n} \sum_{a \in B_\beta} \utility_a^* \leq \beta \log (n+1) \cdot \OPT
		~.
	\]	
	
	Further, recall that $p < 0$.
	We have:
	\[
		\OPT^p = \frac{1}{n} \sum_{a \in A} (\utility_a^*)^{-|p|} \ge \frac{1}{n} \sum_{a \in B_\beta} (\utility_a^*)^{-|p|} 
		~.
	\]
	
	Finally, by H\"{o}lder's inequality:
	\[
		\bigg( \frac{1}{n} \sum_{a \in B_\beta} \utility_a^* \bigg)^{\frac{|p|}{|p|+1}} 
		\bigg( \frac{1}{n} \sum_{a \in B_\beta} (\utility_a^*)^{-|p|} \bigg)^{\frac{1}{|p|+1}} \ge \frac{1}{n} \sum_{a \in B_\beta} 1 = \frac{|B_\beta|}{n}
		~.
	\]
	
	Combining these inequalities proves the lemma.
\end{proof}

\subsection{Proof of \Cref{thm:nashian-to-harmonic}}

We compare the $p$-mean welfare of Nashian Greedy to the optimal benchmark by:
\[
	\bigg( \frac{\ALG}{\OPT} \bigg)^p = \frac{1}{n} \sum_{a\in A} \bigg( \frac{\utility_a}{\OPT} \bigg)^p
	~.
\]

By $\utility_a \ge \nicefrac{1}{n}$ and $\OPT \le 1$, and recalling that $p < 0$, we have :
\[
	\bigg(\frac{\utility_a}{\OPT}\bigg)^p \le n^{|p|}
	~.
\]

Therefore, the above ratio can be written as:
\begin{align*}
	\bigg( \frac{\ALG}{\OPT} \bigg)^p
	& = \int_0^{n^{|p|}} \big(\text{fraction of agents with $\big( \nicefrac{\utility_a}{\OPT} \big)^p \ge \alpha$}\big) \dif{\alpha} \\
	&
	\le \int_0^{n^{|p|}} \big( \alpha^{\frac{1}{p}} \log(n+1) \big)^{\frac{|p|}{|p|+1}} \dif{\alpha} 
	\tag{\Cref{lem:bad-agents-number}} \\
	&
	= \int_0^{n^{|p|}} \big( \log(n+1) \big)^{\frac{|p|}{|p|+1}} \alpha^{-\frac{1}{|p|+1}} \dif{\alpha} \\
	&
	= \frac{|p|+1}{|p|} \cdot n^{\frac{|p|^2}{|p|+1}} \big(\log(n+1) \big)^{\frac{|p|}{|p|+1}}
	~.
\end{align*}

Taking $p$-th root on both sides gives:
\[
	\frac{\ALG}{\OPT} ~\ge \underbrace{\bigg( \frac{|p|}{|p|+1} \bigg)^{\frac{1}{|p|}}}_{\text{$\Omega(1)$ for $|p| = \Omega(1)$}} \cdot~ n^{-\frac{|p|}{|p|+1}} \big(\log(n+1) \big)^{-\frac{1}{|p|+1}}
	~.
\]

\section[From Harmonic to Egalitarian Welfare]{From Harmonic to Egalitarian Welfare: $-\infty \le p \le -1$}
\label{sec:harmonic-to-egalitarian}

This section considers the regime between harmonic welfare and egalitarian welfare. 
In this case, we need to introduce an auxiliary component to obtain a new lower bound of the agents' utilities that is better than the $\nicefrac{1}{n}$ base utility given by Uniform Allocation.

Although Nashian Greedy on its own fails to provide a better bound, we observe that the number of agents who are in trouble cannot exceed $\sqrt{n \log (n+1)}$ (see \Cref{cor:critical-threshold}).
Hence, we may reserve a constant fraction of each item to help these agents.

How can we identify the agents in trouble?
Na\"{i}vely, it is natural to consider an egalitarian greedy algorithm that allocates the item to the agents who currently have the lowest utilities.
However, an agent may have low utility just because its valuable items are yet to arrive.
It would be a mistake to allocate items to such an agent purely based on egalitarian consideration.
Instead, we introduce for each agent $a \in A$ a regularization term $\regularizer_a$, which equals the agent's remaining monopolist utility, i.e., its total utility for the items yet to arrived, scaled by a $\sqrt{n \log (n+1)}$ factor that is driven by the analysis.
Our new component is the egalitarian greedy allocation with respect to the regularized utilities $\utility_a + \regularizer_a$.

We remark that \citet{BarmanKM:AAAI:2022} also introduced an algorithmic sub-routine to identify the agents in trouble, and already observed the necessity to consider both the agents' utilities and their remaining monopolist utilities. 
Our regularized egalitarian allocation gives an alternative, and in our opinion, more natural way to implement this idea.
The resulting competitive ratios are also asymptotically better than those from \citet{BarmanKM:AAAI:2022}.

Below we present the formal definition of this Mixed Greedy algorithm that combines Nashian Greedy and Regularized Egalitarian Greedy.
For cleaner constants in the analysis, we further relax the problem by letting there be two copies of each item.
This relaxation only affects the competitive ratio by a constant factor.

\begin{algorithm}{Mixed Greedy (Nashian Greedy and Regularized Egalitarian Greedy)}
	\emph{Initialization:~}
	Let $x_a(t) = 0$, $\utility_a(0) = \frac{1}{n}$ for any agent $a \in A$ and any time $t \in I$.\\[2ex]
	\emph{Regularization:~}
	Define the regularizer of any agent $a \in A$ as:
	\[
	\regularizer_a(t) = \frac{1}{\Phi} 
	\Big(1 - \int_0^t v_a(s) \dif{s} \Big)
	~,
	\]
	where $\Phi = \sqrt{n \log(n+1)}$.\\[2ex]
	\emph{Online Decisions:~}
	For each item $t \in I$, which comes in two copies:
	\begin{enumerate}
	\item Allocate the \emph{Nashian copy} using Nashian Greedy.
	\item Allocate the \emph{egalitarian copy} to the agents with the smallest regularized utility:
		\[
			\utility_a(t) + \regularizer_a(t)
			~,
		\]
		to greedily maximize the regularized egalitarian welfare:
		\[
			\min_{a \in A} \big( \utility_a(t) + \regularizer_a(t) \big)
			~.
		\]
	\end{enumerate}
\end{algorithm}

This algorithm achieves the same and nearly optimal competitive ratio simultaneously for all values of $p$ from harmonic welfare to egalitarian welfare.

\begin{theorem}
	\label{thm:harmonic-to-egalitarian}
	For any $-\infty \le p \le -1$, Mixed Greedy is $O \big( \sqrt{n \log n} \big)$-competitive for the relaxed online $p$-mean welfare maximization problem.
\end{theorem}

\subsection{Properties of the Regularized Utilities}

\begin{definition}[Critical Agents]
\label{def:critical}
	For any time $t \in I$, and for any $\beta  \geq 0$, let $C_\beta(t)$ be the set of \emph{$\beta$-critical} agents at time $t$, whose regularized utilities are at most a $\beta$ fraction of the optimal $p$-mean welfare, i.e.:
	\[
		C_{\beta}(t) = \Big\{a\in A:\,\utility_a(t)+\regularizer_a(t) \leq \beta \cdot \OPT \Big\}
		~.
	\]
\end{definition}

\begin{lemma}
	\label{lem:critical-agent-number}
	For any time $t\in I$ and any $\beta > 0$, the fraction of $\beta$-critical agents at time $t$ is upper bounded by:
	\[
		\frac{|C_\beta(t)|}{n} \le \max \Big\{ \, \big(2 \beta \log(n+1) \big)^{\frac{|p|}{|p|+1}} \,,\, (2\Phi\beta)^{|p|} \,\Big\}
		~.
	\]
\end{lemma}

\begin{proof}    
	On one hand, by the definition of $p$-mean welfare and $p < 0$, we have:
	\[
		\OPT^p = \frac{1}{n} \sum_{a\in A} \big(\utility_a^*\big)^{-|p|} \ge \frac{1}{n} \sum_{a \in C_\beta(t)} \big(\utility_a^*\big)^{-|p|}
		~.
	\]
	
	On the other hand, all $\beta$-critical agents are $\beta$-bad.
	We have:
	\begin{align*}
		\frac{1}{n} \sum_{a \in C_\beta(t)} \utility_a^*
		&
		\leq \beta \log (n+1) \cdot \OPT + \frac{1}{n} \sum_{a \in C_\beta(t)} \Big( 1 - \int_0^t v_a(s) \dif{s} \Big)
		\tag{\Cref{lem:bad-agents-optimal-utility}} \\
		&
		= \beta \log (n+1) \cdot \OPT + \frac{1}{n} \sum_{a \in C_\beta(t)} \Phi \cdot \regularizer_a(t) \\
		&
		\le \bigg(\beta \log (n+1) + \frac{|C_\beta(t)|}{n} \cdot \Phi \beta \bigg) \cdot \OPT
		~.
		\tag{$\regularizer_a(t) \le \utility_a(t) + \regularizer_a(t) \le \beta \cdot \OPT$}
	\end{align*}
	
	By H\"{o}lder's inequality:
	\[
		\bigg( \frac{1}{n} \sum_{a \in C_\beta(t)} \big(\utility_a^*\big)^{-|p|} \bigg)^{\frac{1}{|p|+1}} \bigg( \frac{1}{n} \sum_{a \in C_\beta(t)} \utility_a^* \bigg)^{\frac{|p|}{|p|+1}} \ge \frac{1}{n} \sum_{a \in C_\beta(t)} 1 = \frac{|C_\beta(t)|}{n}
		~.
	\]
	
	Combining these inequalities gives:
	\[
		\bigg( \beta \log (n+1) + \frac{|C_\beta(t)|}{n} \cdot \Phi \beta \bigg)^{\frac{|p|}{|p|+1}} \ge \frac{|C_\beta(t)|}{n}
		~.
	\]
	
	Rearranging terms, we have:
	\[
		\beta \log (n+1) \cdot \bigg( \frac{n}{|C_\beta(t)|} \bigg)^{\frac{|p|+1}{|p|}} + \Phi\, \beta \cdot \bigg( \frac{n}{|C_\beta(t)|} \bigg)^{\frac{1}{|p|}} \ge 1
		~.
	\]
	
	Hence, either the first part is at least a half, in which case:
	\[
		\frac{|C_\beta(t)|}{n} \le \big(2 \beta \log(n+1) \big)^{\frac{|p|}{|p|+1}}
		~,
	\]
	or the second part is at least a half, in which case:
	\[
		\frac{|C_\beta(t)|}{n} \le (2\Phi\beta)^{|p|}
		~.
	\]
	
	In either case, the lemma follows.
\end{proof}

\begin{corollary}
	\label{cor:critical-threshold}
	For any time $t \in I$, and:
	\begin{equation}
		\label{eq:beta-star}
		\beta^* = \frac{1}{2} \cdot  n^{-\frac{1}{2}-\frac{1}{2|p|}} (\log (n+1))^{-\frac{1}{2}+\frac{1}{2|p|}}
		~,	
	\end{equation}
	we have:
	\[
		\big|C_{\beta^*} (t)\big| \le \sqrt{n \log (n+1)}
		~.
	\]
\end{corollary}

Using \Cref{cor:critical-threshold}, we derive a universal lower bound for all agents' utilities.

\begin{lemma}
	\label{lem:negative-infinity-main-lemma}
	For the choice of $\beta^*$ in \Cref{eq:beta-star}, and any agent $a \in A$, we have:
	\[
		\utility_a \ge \beta^* \cdot \OPT
		~.
	\]
\end{lemma}

\begin{proof}
	We will prove a stronger claim that for any agent $a \in A$ and any time $t \in I$:
	\begin{equation}
		\label{eqn:regularized-utility}
		\utility_a(t) + \regularizer_a(t) \ge \beta^* \cdot \OPT
		~.
	\end{equation}
	
	Then, the lemma holds as the special case when $t = T$ because $\regularizer_a(T) = 0$.

	Initially at time $t = 0$, we have:
	\[
		\regularizer_a(0) = \frac{1}{\Phi} > \beta^* \ge \beta^* \cdot \OPT
		~.
	\]
	
	To prove that \Cref{eqn:regularized-utility} holds at all time $t$, it suffices to show that for any time $t \in I$ when there is at least one $\beta^*$-critical agent $a \in C_{\beta^*}(t)$, the allocation of the egalitarian copy of item $t$ weakly increases the regularized egalitarian welfare.
	
	For any critical agent $a \in C_{\beta^*}(t)$, we have:
	\[
	\frac{\dif}{\dif{t}} \regularizer_a (t) = - \frac{v_a(t)}{\Phi}
	~.
	\] 
	
	Further, \Cref{cor:critical-threshold} asserts that at most $\sqrt{n \log(n+1)}$ agents are $\beta^*$-critical.
	Hence, allocating the egalitarian copy of item $t$ equally among these $\beta^*$-critical agents would have yielded:
	\[
		\frac{\dif}{\dif{t}} \utility_a(t) \ge \frac{v_a(t)}{\sqrt{n \log (n+1)}} = \frac{v_a(t)}{\Phi}
		~,
	\]
	and weakly increased the regularized egalitarian welfare. 
	The greedy allocation of the algorithm would only do better.
\end{proof}

\subsection{Proof of \Cref{thm:harmonic-to-egalitarian}}

The proof is almost verbatim to that of \Cref{thm:nashian-to-harmonic}, except that we will use the newly developed \Cref{lem:negative-infinity-main-lemma} to lower bound the agents' utilities, replacing the basic bound $\utility_a \ge \nicefrac{1}{n}$.

Consider the $p$-th power of the Mixed Greedy algorithm's $p$-mean welfare, normalized by the $p$-th power of $\OPT$:
\[
	\bigg( \frac{\ALG}{\OPT} \bigg)^p = \frac{1}{n} \sum_{a\in A} \bigg( \frac{\utility_a}{\OPT} \bigg)^p
	~.
\]

By \Cref{lem:negative-infinity-main-lemma} and $p < 0$, we have:
\[
	\bigg(\frac{\utility_a}{\OPT}\bigg)^p \le (\beta^*)^p
	~.
\]

The above ratio can therefore be written as:
\begin{align*}
	\bigg( \frac{\ALG}{\OPT} \bigg)^p
	& = \int_0^{(\beta^*)^p} (\text{fraction of agents with $\big( \nicefrac{\utility_a}{\OPT} \big)^p \ge \alpha$)} \dif{\alpha} \\
	&
	\le \int_0^{(\beta^*)^p} \big( \alpha^{\frac{1}{p}} \log(n+1) \big)^{\frac{|p|}{|p|+1}} \dif{\alpha} 
	\tag{\Cref{lem:bad-agents-number}} \\
	&
	= \int_0^{(\beta^*)^p} \big( \log(n+1) \big)^{\frac{|p|}{|p|+1}} \alpha^{-\frac{1}{|p|+1}} \dif{\alpha} \\
	&
	= \frac{|p|+1}{|p|} (\beta^*)^{-\frac{|p|^2}{|p|+1}} \big(\log(n+1) \big)^{\frac{|p|}{|p|+1}} \\
	&
	= \frac{|p|+1}{|p|} 2^{{\frac{|p|^2}{|p|+1}}} \big( n \log(n+1) \big)^{\frac{|p|}{2}}
	~.
\end{align*}

Taking $p$-th root on both sides gives:
\[
	\frac{\ALG}{\OPT} ~\ge~ \underbrace{\bigg( \frac{|p|}{|p|+1} \bigg)^{\frac{1}{|p|}} 2^{-{\frac{|p|}{|p|+1}}}}_{\text{$\Omega(1)$ for $p \le -1$}} ~\cdot~ {\Big(\sqrt{n \log(n+1)}\Big)}^{-1}
	~.
\]

\section[Hardness for the Nashian and Positive Regimes]{Hardness for the Nashian and Positive Regimes: $p \ge -\nicefrac{1}{\log n}$}

Recall that the online primal dual algorithm by \citet{DevanurJ:STOC:2012} is $\nicefrac{1}{p}$\,-competitive for $\nicefrac{1}{\log n} \le p \le 1$ (\Cref{thm:p>0-1/p}), and Nashian Greedy is $O(\log n)$-competitive for $- \nicefrac{1}{\log n} \le p \le \nicefrac{1}{\log n}$ (\Cref{cor:nashian}).
This section will complement these competitive ratios with hardness results that match them up to a lower-order term.
Formally, we will show that:

\begin{theorem}
    \label{thm:hardness-positive}
    For any $p \le 1$ and any online algorithm for online $p$-mean welfare maximization, the competitive ratio is no smaller than:
    \[
    	\begin{cases}
    	    \Omega \big( \frac{1}{p} \big) & \mbox{if $p \ge \frac{\log\log n}{\log n}$;} \\[1ex]
            \Omega \big( \frac{\log n}{\log\log n} \big) & \mbox{otherwise,}
    	\end{cases}
    \]
    even when the agents have binary valuations (and unit monopolist utilities).
\end{theorem}

This asymptotic lower bound holds for a sufficiently large number of agents.
It implicitly covers three cases. 
If $p$ is a positive constant independent of the number of agents $n$, it characterizes how the competitive ratio changes as $p$ tends to $0$.
If $p(n)$ is a function of the number of agents $n$ that converges to zero as $n$ goes to infinity, e.g., $p(n) = \nicefrac{1}{\log n}$, the theorem asserts that for the competitive ratio is no smaller than $\Omega(\nicefrac{1}{p(n)})$ if $p(n)$ converges to zero slower than function $\nicefrac{\log\log n}{\log n}$, or no smaller than $\Omega(\nicefrac{\log n}{\log\log n})$ otherwise.
The latter bound also applies when $p$ is a negative constant or a negative function $p(n)$ that converges to zero.

Since the items are divisible, we may without loss of generality focus on deterministic online algorithms.
Given any randomized online algorithm, we may instead consider a corresponding deterministic algorithm that allocates each item $i \in I$ to each agent $a \in A$ by the expected amount allocated by the randomized algorithm.

\subsection{Hard Instance}

It suffices to consider $p \le \nicefrac{1}{16}$ and a sufficiently large $n$ to prove the asymptotic bounds in \Cref{thm:hardness-positive}. 
The hard instance depends on two parameters:
\[
	M =
        \begin{cases}
        \frac{1}{4p} & \mbox{if $p \geq \frac{\log\log n}{\log n}$;}\\[1ex]
        \frac{\log n}{4 \log \log n} & \mbox{otherwise.}
        \end{cases}
	\quad ,\quad
	L = \frac{\log n}{2} \ge 2M \log \log n
	~.
\]
By $p \le \nicefrac{1}{16}$ and for a sufficiently large $n$, we have $M \ge 4$.

The items arrive in two stages: an \emph{Upper Triangular Stage} and a \emph{Makeup Stage}.
See \Cref{fig:hardness-nashian} for an illustration of this hard instance.

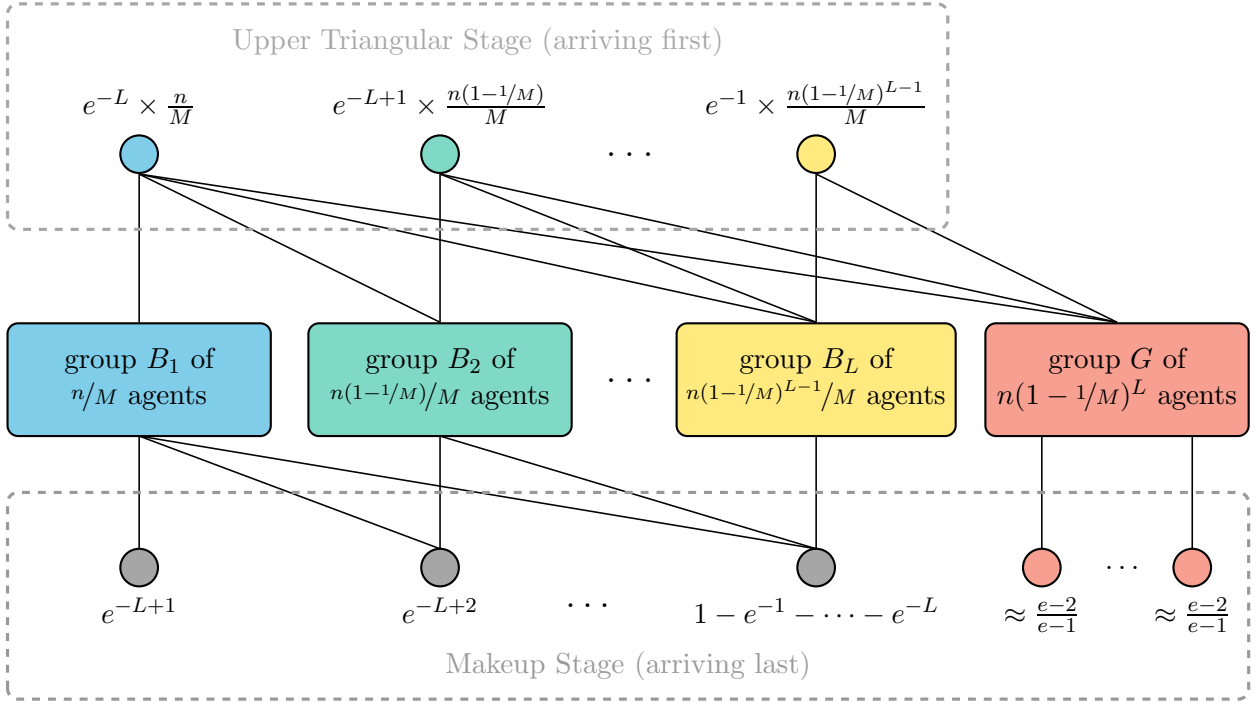
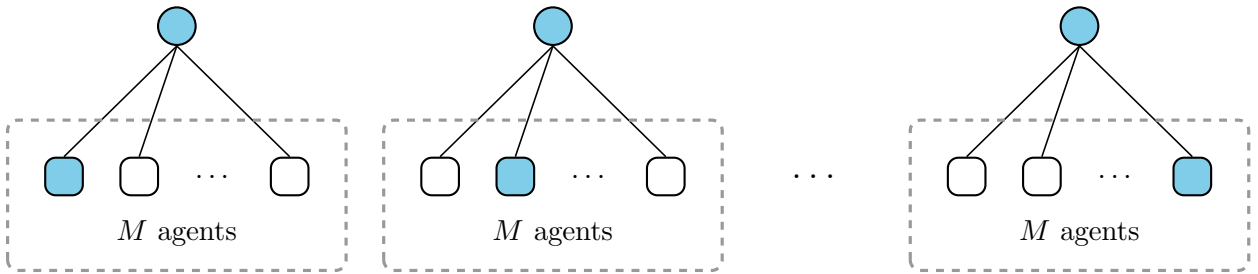
\begin{figure}[htp]
\centering
\begin{subfigure}{\textwidth}
\centering
\begin{tikzpicture}[
		group/.style = {
			draw=black,
			thick,
			rounded corners,
			minimum width=3.5cm,
			minimum height=1.5cm,
			align=center,
		},
		item/.style = {
			minimum size=.5cm,
			circle,
			draw=black,
			thick,
		},
		like/.style = {
			draw=black,
			semithick,
		}
	]
	\node[group,fill=hkublue!50] (B1) at (0,0) {group $B_1$ of\\ $\nicefrac{n}{M}$ agents};
	\node[group,fill=hkugreen!50] (B2) at (4,0) {group $B_2$ of\\ $\nicefrac{n (1-\nicefrac{1}{M})}{M}$ agents};
	\node at (6.57,0) {\LARGE\dots};
	\node[group,fill=hkuyellow!50] (Bk) at (9,0) {group $B_{L}$ of\\ $\nicefrac{n (1-\nicefrac{1}{M})^{L-1}}{M}$ agents};	
	\node[group,fill=hkured!50] (G) at (13,0) {group $G$ of\\ $n(1-\nicefrac{1}{M})^L$ agents};
	\node[item,label=above:{$e^{-L} \times \frac{n}{M}$},fill=hkublue!50] (T1) at (0,3) {};
	\node[item,label=above:{$e^{-L+1} \times \frac{n(1-\nicefrac{1}{M})}{M}$},fill=hkugreen!50] (T2) at (4,3) {};

    \node at (6.57,3) {\LARGE\dots};
	\node[item,label=above:{$e^{-1} \times\frac{n(1-\nicefrac{1}{M})^{L-1}}{M}$},fill=hkuyellow!50] (Tk) at (9,3) {};
	\draw[like] (T1.south)--(B1.north);
	\draw[like] (T1.south)--(B2.north);
	\draw[like] (T1.south)--(Bk.north);
	\draw[like] (T1.south)--(G.north);	
	\draw[like] (T2.south)--(B2.north);
	\draw[like] (T2.south)--(Bk.north);
	\draw[like] (T2.south)--(G.north);
	\draw[like] (Tk.south)--(Bk.north);
	\draw[like] (Tk.south)--(G.north);
	\draw[very thick,draw=gray!70,dashed,rounded corners] (-1.75,2) rectangle (10.75,5);
	\node[align=center] at (4.5,4.5) {\textcolor{gray!70}{Upper Triangular Stage (arriving first)}};
	\node[item,label=below:{$\approx \frac{e-2}{e-1}$},fill=hkured!50] at (12,-2.5) {};
	\node at (13.1,-2.5) {\dots};
	\node[item,label=below:{$\approx \frac{e-2}{e-1}$},fill=hkured!50] at (14,-2.5) {};
	\draw[like] (12,-2.25) -- (12,-0.75);
	\draw[like] (14,-2.25) -- (14,-0.75);
	\node[item,label=below:{$e^{-L+1}$},fill=gray!70] at (0,-2.5) {};
	\draw[like] (0,-2.25) -- (0,-.75);
	\node[item,label=below:{$e^{-L+2}$},fill=gray!70] at (4,-2.5) {};
	\draw[like] (4,-2.25) -- (0,-.75);
	\draw[like] (4,-2.25) -- (4,-.75);
	\node at (6,-3) {\Large\dots};
	\node[item,label=below:{$1-e^{-1}-\dots-e^{-L}$},fill=gray!70] at (9,-2.5) {};
	\draw[like] (9,-2.25) -- (0,-.75);
	\draw[like] (9,-2.25) -- (4,-.75);
	\draw[like] (9,-2.25) -- (9,-.75);
	\draw[very thick,draw=gray!80,dashed,rounded corners] (-1.75,-4.25) rectangle (14.75,-1.5);
	\node at (6.5,-3.8) {\textcolor{gray!80}{Makeup Stage (arriving last)}};
\end{tikzpicture}
\caption{%
	Structure of the hard instance.
	Rounded rectangles represent groups of agents.
	Circles represent items, labeled by their supplies.
	Having an edge between a group of agents and an item means that the agents have value $1$ for the item;
    the parallel edges between the good agents in group $G$ and the corresponding items in the Makeup Stage indicate having a separate item for every agent in the group.
	The items in the Upper Triangular Stage arrive from left to right; each item in the picture corresponds to multiple items in the gadget illustrated in part (b).
	The items' arrival order in the Makeup Stage does not affect the argument.
	The groups of agents and items are colored by the offline optimal allocation;
	some items are colored gray because their allocations do not affect the argument.
}
\end{subfigure}

\vspace{20pt}

\begin{subfigure}{\textwidth}
\centering
\begin{tikzpicture}[
		agent/.style = {
			draw=black,
			thick,
			rounded corners,
			minimum size=.5cm,
			align=center,
		},
		item/.style = {
			minimum size=.5cm,
			circle,
			draw=black,
			thick,
		},
		like/.style = {
			draw=black,
			semithick,
		}
	]
	\node[agent,fill=hkublue!50] (A11) at (0,0) {};
	\node[agent] (A12) at (1,0) {};
	\node at (2,0) {\dots};
	\node[agent] (A13) at (3,0) {};
	\node[agent] (A21) at (5,0) {};
	\node[agent,fill=hkublue!50] (A22) at (6,0) {};
		\node at (7,0) {\dots};
	\node[agent] (A23) at (8,0) {};	
	\node at (10,0) {\Large \dots};
	\node[agent] (A31) at (12,0) {};
	\node[agent] (A32) at (13,0) {};
	\node at (14,0) {\dots};
	\node[agent,fill=hkublue!50] (A33) at (15,0) {};
	\node[item,fill=hkublue!50] (I1) at (1.5,2) {};
	\draw[like] (I1.south) -- (A11.north);
	\draw[like] (I1.south) -- (A12.north);
	\draw[like] (I1.south) -- (A13.north);
	\draw[very thick,draw=gray!80,dashed,rounded corners] (-.75,-1.25) rectangle (3.75,.75);
	\node at (1.5,-.75) {$M$ agents};
	\node[item,fill=hkublue!50] (I2) at (6.5,2) {};
	\draw[like] (I2.south) -- (A21.north);
	\draw[like] (I2.south) -- (A22.north);
	\draw[like] (I2.south) -- (A23.north);	
	\draw[very thick,draw=gray!80,dashed,rounded corners] (4.25,-1.25) rectangle (8.75,.75);
	\node at (6.5,-.75) {$M$ agents};	
	\node[item,fill=hkublue!50] (I3) at (13.5,2) {};
	\draw[like] (I3.south) -- (A31.north);
	\draw[like] (I3.south) -- (A32.north);
	\draw[like] (I3.south) -- (A33.north);		
	\draw[very thick,draw=gray!80,dashed,rounded corners] (11.25,-1.25) rectangle (15.75,.75);
	\node at (13.5,-.75) {$M$ agents};	
\end{tikzpicture}
\caption{%
	Gadget for a single round in the Upper Triangular Stage.
	The (currently) ungrouped agents are partitioned into subsets of size $M$, represented by dashed rounded rectangles.
	For each subset in the partition, there is an item of value $1$ only to the agents in the subset; 
	the item's supply is $v_\ell = e^{-L-1+\ell}$ in round $\ell$.
	The colored agent in each subset is an agent with the lowest utility after the  allocation of this item.
	These colored agents will form the group of bad agents from this round.
}
\end{subfigure}
\caption{%
	Illustration of the hard instance for Nashian and positive regimes, i.e., when $p \ge -\nicefrac{1}{\log n}$.
	}
\label{fig:hardness-nashian}
\end{figure}

\paragraph{Upper Triangular Stage.}
This stage proceeds in $L$ rounds, during which we will adaptively partition the agents into $L+1$ groups.
After each round, we will determine one group of agents in the partition.
We will refer to the agents in these $L$ groups as the \emph{bad agents} because their utilities for the algorithm's allocation will be at most an $O(\nicefrac{1}{M})$ fraction of their utilities for the optimal allocation.
After these $L$ rounds, we will put all remaining agents into the last group, and call them the \emph{good agents} because their utilities for the algorithm's allocation will be greater than or equal to their utilities for the optimal allocation.
However, the fraction of good agents will be small.

We now describe the Upper Triangular Stage in detail.
In the first round, morally speaking, we would like to have one item with supply:
\[
	v_1 \cdot \frac{n}{M}
	\quad
	\text{where}
	\quad
	v_1 = e^{-L} = \frac{1}{\sqrt{n}}
	~,
\]
for which all agents have value $1$.
However, this would violate the assumption of unit monopolist utilities. %
Instead, we use the following gadget from the hard instance by \citet{BanerjeeGGJ:SODA:2022}.
We partition the agents into subsets of $M$ agents arbitrarily.
For each subset $S$ in the partition, let there be an item with supply $v_1$ that  is of value $1$ to the agents in $S$, and $0$ to the others.
Note that the items' total value remains $v_1 \cdot \nicefrac{n}{M}$, as long as we allocate every item to an agent with value $1$ for it.
Meanwhile, the contribution to each agent's monopolist utility reduces to $v_1$.

Since the instance so far is symmetric for all agents, intuitively the online algorithm should simply distribute every item equally to the agents who have value $1$ for it.
In that case, the agents' utilities would be $\nicefrac{v_1}{M}$ and we may let $B_1$ be an arbitrary subset of $\nicefrac{n}{M}$ agents.
In general, we let $B_1$ be comprised of an agent with the \emph{lowest} utility from each subset in the partition.

We repeat this process $L$ times to adaptively define $L$ groups of bad agents $B_1, B_2, \dots, B_L$.
At the beginning of round $1 \le \ell \le L$, there are:
\[
	n_\ell ~\defeq~ n \Big(1 - \frac{1}{M} \Big)^{\ell-1}
\]
agents not in groups $B_1, \dots, B_{\ell - 1}$.
Morally speaking, we would like to have one item with supply:
\[
	v_\ell \cdot \frac{n_\ell}{M}
	\quad
	\text{where}
	\quad
	v_\ell ~\defeq~ e^{-L-1+\ell}
	~,
\]
that is of value $1$ to these $n_\ell$ ungrouped agents, and $0$ to the others.
To reduce the contribution to the agents' monopolist utilities, we partition the ungrouped agents into subsets of $M$ agents arbitrarily.
For each subset $S$ in the partition, let there be an item with supply $v_\ell$ that is of value $1$ to the agents in $S$, and $0$ to the others.
After the allocation, we let $B_\ell$ be comprised of an agent with the lowest utility from each subset in the partition.

After all $L$ rounds of the Upper Triangular Stage, there are still:
\[
	n \Big(1 - \frac{1}{M}\Big)^L
\]
ungrouped agents.
Let $G$ denote this set of agents, which we will refer to as the good agents.

\paragraph{Makeup Stage.}
Observe that the agents' monopolist utilities are still smaller than $1$ after the Upper Triangular Stage. 
Next, we let additional items arrive in the Makeup Stage to achieve unit monopolist utilities.
For the good agents, the algorithm's allocation yields higher utilities than the offline benchmark.
Hence, we will help the offline benchmark by letting each good agent have its own item in the Makeup Stage, so that its utility for the offline benchmark is at least $\Omega(1)$.
For the bad agents, by contrast, we will let them compete for the same items, in order to limit the online algorithm's ability to catch up with the offline benchmark in terms of the bad agents' utilities. 

We now describe the Makeup Stage in detail.
For each good agent $a \in G$, the Makeup Stage has an item with supply: %
\begin{equation}
	\label{eqn:nashian-hardness-good-agent}
	1 - \sum_{\ell=1}^L v_\ell ~=~ 1 - \frac{1}{e} - \frac{1}{e^2} - \dots - \frac{1}{e^L} ~\ge~ \frac{e-2}{e-1}
\end{equation}
and is of value $1$ only to agent $a$.

Finally, we consider the bad agents.
For each group $B_\ell$, the agents therein have monopolist utilities $\sum_{i=1}^\ell v_i$ after the Upper Triangular Stage.
Let there be an item with supply $1 - \sum_{\ell=1}^L v_\ell$ that is of value $1$ to all bad agents, and $0$ to the good agents.
Further, for any $1 < \ell \le L$, let there be an item with supply $v_\ell$ that is of value $1$ to the agents in $B_1$ to $B_{\ell-1}$, and $0$ to the others.
By the construction of these items, all bad agents have unit monopolist utilities at the end.

\subsection{Properties of the Hard Instance}

First, we describe an offline allocation for the above hard instance, and lower bound the agents' utilities for it.
These bounds for the agents' utilities will provide a lower bound for the optimal $p$-mean welfare in our proof of \Cref{thm:hardness-positive}.

\begin{lemma}
	\label{lem:hardness-nashian-optimal}
	There is an allocation such that for any $1 \le \ell \le L$, the agents in group $B_\ell$ have utilities at least $v_\ell$, and the agents in group $G$ have utilities at least $\nicefrac{(e-2)}{(e-1)}$.
\end{lemma}

\begin{proof}
	In the Makeup Stage, we let each good agent get the unique item for which it has value $1$.
	Then, the utility bound for the good agents follows by \Cref{eqn:nashian-hardness-good-agent}.
	
	Next, we consider the bad agents' utilities. 
	We allocate every item in round $1 \le \ell \le L$ of the Upper Triangular Stage to the unique agent in $B_\ell$ who has value $1$ for it.
	Then, every agent in $B_\ell$ gets utility $v_\ell$.
\end{proof}

Further, we characterize the limitations of the allocation by any deterministic online algorithm.
For any $1 \le \ell \le L$, let $\bar{\utility}_\ell$ denote the average utility of the agents in group $B_\ell$.
We show that these bad agents on average only receive an $O(\nicefrac{1}{M})$ fraction of what they get in \Cref{lem:hardness-nashian-optimal}.

\begin{lemma}
	\label{lem:hardness-nashian-algorithm}
	For any $1 \le \ell \le L$, we have $\bar{\utility}_\ell \le \nicefrac{3 v_\ell}{M}$.
\end{lemma}

\begin{proof}
	For the agents in group $B_\ell$, it suffices to consider the items that arrived in the first $\ell$ rounds of the Upper Triangular Stage and the items in the Makeup Stage, because they have zero utilities for the remaining items.
	
	The total value of the items in the first $\ell$ rounds of the Upper Triangular Stage equals:
	\begin{align*}
		\sum_{i=1}^\ell \frac{n_i}{M} \cdot v_i
		&
		~=~ \frac{n_\ell\, v_\ell}{M} \,\cdot\, \sum_{i=1}^\ell \Big( \frac{M}{e(M-1)} \Big)^{\ell - i} \\
		&
		~\le~ \frac{n_\ell\, v_\ell}{M} \,\cdot\, \sum_{i=1}^\ell 2^{-\ell+i} 
		\tag{$M \ge 4$} \\
		&
		~\le~ \frac{2 \, n_\ell \, v_\ell}{M}
		~.
	\end{align*}
	
	By definition, $B_\ell$ is comprised of the $\nicefrac{n_\ell}{M}$ agents with the lowest utilities after round $\ell$, among the $n_\ell$ agents who are not in groups $1$ to $\ell - 1$.
	Therefore, their average utility from the items in the first $\ell$ rounds of the Upper Triangular Stage is at most:
	\[
	 	\frac{1}{n_\ell} \cdot \frac{2 \, n_\ell \, v_\ell}{M} = \frac{2\, v_\ell}{M}
		~.
	\]
	
	The items in the Makeup Stage are of total value less than $1$ to the $\nicefrac{n_\ell}{M}$ agents in $B_\ell$.
	Hence, their average utility from it is at most:
	\[
		\frac{M}{n_\ell}
		\le \frac{n_\ell v_\ell}{M}
		~,
	\]
	where the inequality follows by $v_\ell \ge v_1 = \nicefrac{1}{\sqrt{n}}$, 
    $M\geq 4$, and for sufficiently large $n$:
	\[
		n_\ell \ge n \Big(1-\frac{1}{M}\Big)^L \ge n \Big(1-\frac{1}{4}\Big)^L > \sqrt{n}
		~.
	\]

	Combining these two parts proves the lemma.
\end{proof}

\subsection{Proof of \Cref{thm:hardness-positive}}

By \Cref{lem:hardness-nashian-optimal}, we have:
\begin{align*}
	\OPT
	&
	~\ge~ 
	\left( \Big(1-\frac{1}{M}\Big)^L \cdot \Big(\frac{e-2}{e-1}\Big)^p ~+~ \sum_{\ell=1}^L \frac{1}{M} \Big(1 - \frac{1}{M}\Big)^{\ell-1} \cdot v_\ell^p \right)^{\frac{1}{p}} \\
	&
	~\ge~
	\frac{e-2}{e-1}
	\left( \Big(1-\frac{1}{M}\Big)^L ~+~ \sum_{\ell=1}^L \frac{1}{M} \Big(1 - \frac{1}{M}\Big)^{\ell-1} \cdot v_\ell^p \right)^{\frac{1}{p}}
	~.
\end{align*}

By the concavity of $p$-mean welfare in the agents' utilities and \Cref{lem:hardness-nashian-algorithm}, we have:
\begin{align*}
	\ALG
	&
	~\le~
	\left( \Big(1-\frac{1}{M}\Big)^L \cdot 1^p ~+~ \sum_{\ell=1}^L \frac{1}{M} \Big(1 - \frac{1}{M}\Big)^{\ell-1} \cdot \bar{\utility}_\ell^p \right)^{\frac{1}{p}} \\
	&
	~\le~
	\left( \Big(1-\frac{1}{M}\Big)^L ~+~ \sum_{\ell=1}^L \frac{1}{M} \Big(1 - \frac{1}{M}\Big)^{\ell-1} \cdot \Big( \frac{3v_\ell}{M} \Big)^p \right)^{\frac{1}{p}}
	~.
\end{align*}

To simplify notations, let:
\[
	A = \Big( 1 - \frac{1}{M} \Big)^L
	\quad,\quad
	B(p) = \sum_{\ell=1}^L \frac{1}{M} \Big(1 - \frac{1}{M}\Big)^{\ell-1} v_\ell^p
	~.
\]

We get that:
\[
	\frac{\ALG}{\OPT} \le \frac{e-1}{e-2}  
	\left( 
		\frac{ A + B(p) \cdot \big( \frac{3}{M} \big)^p}{A + B(p)} 
	\right)^{\frac{1}{p}}
\]

Next, we bound the ratio between $A$ and $B(p)$.
Recall that $v_\ell = e^{-L-1+\ell}$. 
We have:
\[
	\frac{B(p)}{A} = \frac{1}{M} \sum_{\ell=1}^L \left(\Big(1-\frac{1}{M}\Big) e^p\right)^{-(L+1-\ell)}
	~.
\]

By $1-x \le e^{-x}$ and $p \le \nicefrac{1}{2M}$, we have:
\[
	\Big(1 - \frac{1}{M}\Big) e^p 
	~\le~ 
	e^{-\frac{1}{M}+p}
	~\le~
	e^{-\frac{1}{2M}} 
	~.
\]

Hence, we get that:
\[
	\frac{B(p)}{A}
	~\ge~ 
	\frac{1}{M} \sum_{\ell=1}^L e^{\frac{L+1-\ell}{2M}}
	~=~ 
	\frac{e^\frac{L}{2M}-1}{M(1-e^{-\frac{1}{2M}})} 
	~.
\]

On one hand, we have:
\[
	1-e^{-\frac{1}{2M}} \le \frac{1}{2M}
	~.
\]

On the other hand, recall that $L \ge 2M \log\log n$.
We get that:
\[
	e^\frac{L}{2M} \ge \log n
\]

Therefore, for sufficiently large $n$, we conclude that:
\[
	\frac{B(p)}{A} \ge 2 \big( \log n - 1 \big) \ge \log n
	~.
\]

It remains to bound:
\begin{equation}
	\label{eqn:hardness-positive-checkpoint}
	\left( \frac{A}{A+B(p)} + \frac{B(p)}{A+B(p)} \Big(\frac{3}{M}\Big)^p \right)^{\frac{1}{p}}
	\le 
	\left( \frac{1}{1+\log n} + \frac{\log n}{1+\log n} \Big( \frac{3}{M} \Big)^p \right)^{\frac{1}{p}}
	~.
\end{equation}

Consider any $p \ge \nicefrac{\log\log n}{\log n}$.
The right-hand-side of \Cref{eqn:hardness-positive-checkpoint} equals:
\begin{align*}
	\frac{3}{M} \left( 1 + \frac{1}{1+\log n} \Big( \Big( \frac{M}{3} \Big)^p - 1 \Big) \right)^{\frac{1}{p}}
	~.
\end{align*}

Further, we have:
\begin{align*}
	\Big( \frac{M}{3} \Big)^p 
	&
	~=~ e^{p \log \frac{M}{3}} \\
	&
	~\le~ e^{p \log \frac{1}{p}} 
	\tag{$M = \nicefrac{1}{4p}$} \\
	&
	~\le~ 1 + 2p\log \frac{1}{p}
	~.
	\tag{$e^x \le 1+2x$ for $0 \le x \le 1$}
\end{align*}

Hence, \Cref{eqn:hardness-positive-checkpoint} is at most:
\begin{align*}
	\frac{3}{M} \left( 1 + \frac{1}{1+\log n} 2p\log \frac{1}{p} \right)^{\frac{1}{p}}
	&
	\le
	\frac{3}{M} \big( 1 + p \big)^\frac{1}{p} 
	\tag{$p \ge \frac{1}{\log n}$} \\
	&
	\le \frac{3e}{M}
	~.
\end{align*}

For $p < \nicefrac{\log\log n}{\log n}$, $M$ does not depend on $p$.
Hence, the right-hand-side of \Cref{eqn:hardness-positive-checkpoint} can be seen as a weighted $p$-mean of $1$ and $\nicefrac{3}{M}$.
This is increasing in $p$.
Hence, it reduces to the case of $p = \nicefrac{\log\log n}{\log n}$.

\section[Hardness for the Negative Regime]{Hardness for the Negative Regime: $p \le - \nicefrac{1}{\log n}$}

Recall that Nashian Greedy is $n^{|p|}$-competitive for $-o(1) \le p \le -\nicefrac{1}{\log n}$ and $n^{\nicefrac{|p|}{(|p|+1)}}$-competitive for $-1 \le p \le -\Omega(1)$, up to lower-order and logarithmic factors (\Cref{thm:nashian}).
Moreover, Mixed Greedy is $O(\sqrt{n \log n})$-competitive for $-\infty \le p \le -1$ (\Cref{thm:harmonic-to-egalitarian}).
This section complements these algorithmic results with almost matching lower bounds.

\subsection{Family of Hard Instances}

We consider a family of hard instances $I(L,\alpha)$ parameterized by a positive integer $L$ and a real number $0 \le \alpha < |p|$.
The items arrive in two stages: an Upper Triangular Stage and a Makeup Stage.
See \Cref{fig:hardness-negative} for an illustration of this family of hard instances.

\bigskip

\begin{figure}[ht]
\centering
\begin{tikzpicture}[
		group/.style = {
			draw=black,
			thick,
			rounded corners,
			minimum width=3.5cm,
			minimum height=1.5cm,
			align=center,
		},
		item/.style = {
			minimum size=.5cm,
			circle,
			draw=black,
			thick,
			fill=hkured!50,
		},
		like/.style = {
			draw=black,
			semithick,
		}
	]
	\node[group,fill=hkublue!50] (G1) at (0,0) {group $G_1$ of\\ $n-n^{s_1}$ agents};
	\node[group,fill=hkugreen!50] (G2) at (4,0) {group $G_2$ of\\ $n^{s_1}-n^{s_2}$ agents};
	\node at (6.57,0) {\LARGE\dots};
	\node[group,fill=hkuyellow!50] (GL) at (9,0) {group $G_L$ of\\ $n^{s_{L-1}}-n^{s_L}$ agents};	
	\node[group,fill=hkured!50] (B) at (13,0) {group $B$ of\\ $n^{s_L}$ agents};
	\node[item,label=above:{$1-n^{s_1-1}$}] (T1) at (0,3) {};
	\node[item,label=above:{$n^{s_1-1}-n^{s_2-1}$}] (T2) at (4,3) {};
	
    \node at (6.57,3) {\LARGE\dots};
    
    \node[item,label=above:{$n^{s_{L-1}-1}-n^{s_L-1}$}] (Tk-1) at (9,3) {};
	\draw[like] (T1.south)--(G1.north);
	\draw[like] (T1.south)--(G2.north);
	\draw[like] (T1.south)--(GL.north);
	\draw[like] (T1.south)--(B.north);	
	\draw[like] (T2.south)--(G2.north);
	\draw[like] (T2.south)--(GL.north);
	\draw[like] (T2.south)--(B.north);
	\draw[like] (Tk-1.south)--(GL.north);
	\draw[like] (Tk-1.south)--(B.north);
	\draw[very thick,draw=hkured!50,dashed,rounded corners] (-1.75,2) rectangle (10.75,5);
	\node at (4.5,4.5) {\textcolor{hkured!50}{Upper Triangular Stage (arriving first)}};
	\node[item,label=below:{$n^{s_1-1}$},fill=hkublue!50] at (-1,-2) {};
	\draw[like] (-1,-1.75) -- (-1,-.75);	
	\node[item,label=below:{$n^{s_1-1}$},fill=hkublue!50] at (1,-2) {};
	\draw[like] (1,-1.75) -- (1,-.75);	
	\node at (0,-1.75) {\large\dots};
	\node[item,label=below:{$n^{s_2-1}$},fill=hkugreen!50] at (3,-2) {};
	\draw[like] (3,-1.75) -- (3,-.75);	
	\node[item,label=below:{$n^{s_2-1}$},fill=hkugreen!50] at (5,-2) {};
	\draw[like] (5,-1.75) -- (5,-.75);	
	\node at (4,-1.75) {\large\dots};	
	\node[item,label=below:{$n^{s_L-1}$},fill=hkuyellow!50] at (8,-2) {};
	\draw[like] (8,-1.75) -- (8,-.75);	
	\node[item,label=below:{$n^{s_L-1}$},fill=hkuyellow!50] at (10,-2) {};
	\draw[like] (10,-1.75) -- (10,-.75);	
	\node at (9,-1.75) {\large\dots};
	\node[item,label=below:{$n^{s_L-1}$}] (TB) at (13,-2) {};	
	\draw[like] (TB.north)--(B.south);		
	\draw[very thick,draw=gray!80,dashed,rounded corners] (-1.75,-4) rectangle (14.75,-1.25);
	\node at (6.5,-3.5) {\textcolor{gray!80}{Makeup Stage (arriving last)}};
\end{tikzpicture}	
\caption{%
	Illustration of the family of hard instances for the negative regime. 
	Here, $L \ge 1$ is the number of groups of good agents.
	A decreasing sequence $1=s_0>s_1>\dots>s_L\geq 0$ determines the group sizes. 
	Rounded rectangles represent groups of agents, labeled by the groups' indices and sizes.
	Circles represent items, labeled by the items' supplies.
	Having an edge between a group of agents and an item means that the agents have value $1$ for the item;
	the parallel edges between the good agents and the corresponding items in the Makeup Stage indicate having a separate item for every such agent.	
	The items in the Upper Triangular Stage arrive first from left to right;
	then, those in the Makeup Stage arrive by an arbitrary order.
	The groups of agents and items are colored according to the offline optimal allocation.
	For example, the items in the Upper Triangular Stage are colored red because the offline benchmark allocates them to group $B$ which is also colored red.
	}
\label{fig:hardness-negative}
\end{figure}
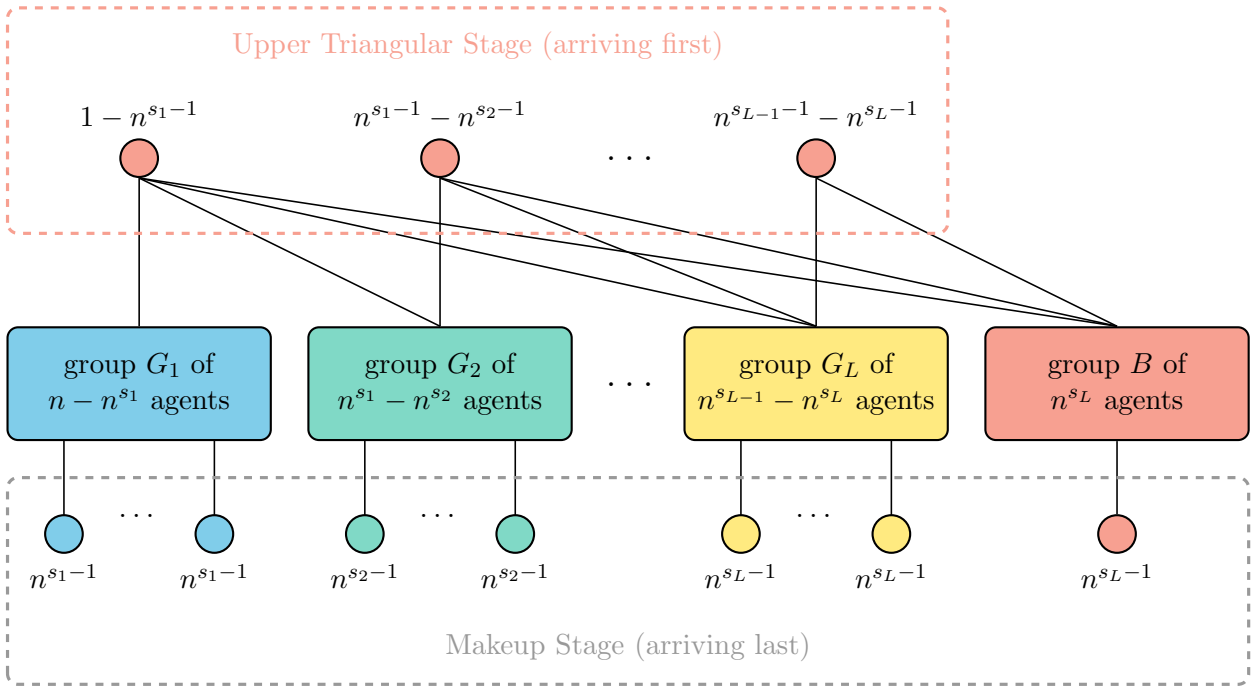

\paragraph{Upper Triangular Stage.}
This stage proceeds in $L$ rounds, during which we will adaptively partition the agents into $L+1$ groups.
After each round, we will determine one group of agents in the partition.
We will refer to the agents in these $L$ groups as the \emph{good agents} because their utilities for the algorithm's allocation might be greater than their utilities for the optimal allocation.
After these $L$ rounds, we will put all remaining agents into the last group, and call them the \emph{bad agents} because their utilities for the algorithm's allocation will be low.

We now describe the Upper Triangular Stage in detail.
Based on parameters $L$ and $\alpha$, we will choose a decreasing sequence:
\[
	1 = s_0 > s_1 > \dots > s_L \geq 0
	~,
\] 
such that the $\ell$-th group of good agents $G_\ell$ has size $n^{s_{\ell-1}} - n^{s_\ell}$.
Correspondingly, the number of bad agents in $B$ equals $n^{s_L}$.
We will first present the structure of the hard instance, deferring the choice of this decreasing sequence to the end of the subsection.

In the first round, let an item arrive with supply $n^{s_0-1} - n^{s_1-1}$ and value $1$ to all agents.
Given the online algorithm's allocation of this item, let $G_1$ be comprised of $n^{s_0} - n^{s_1}$ agents with the \emph{highest} utilities (breaking ties arbitrarily), so that the remaining agents' average utility is at most:
\[
	\frac{1}{n} \cdot \big( n^{s_0-1} - n^{s_1-1} \big) < \frac{1}{n}
	~.
\]

Here, recall that we may consider deterministic online algorithms without loss of generality.

We repeat this process to adaptively define $L$ groups of good agents $G_1, G_2, \dots, G_L$.
In each round $1 \le \ell \le L$, let an item arrive with supply $n^{s_{\ell-1}-1} -n^{s_\ell-1}$ and value $1$ to the ungrouped agents, i.e., those not in $G_1$ to $G_{\ell-1}$.
Given the online algorithm's allocation of this item, let the $G_\ell$ be comprised of the $n^{s_{\ell-1}} - n^{s_\ell}$ ungrouped agents with the highest utilities.

After all $L$ rounds of the Upper Triangular Stage, there are still $n^{s_L}$ ungrouped agents.
Let $B$ denote this set of agents, which we will refer to as the bad agents.

\paragraph{Makeup Stage.}
We next let additional items arrive such that the agents have unit monopolist utilities at the end.
Since the good agents might have high utilities for the algorithm's allocation, we will help the offline benchmark by letting each good agent have its own item in the Makeup Stage, so that the agent's utility for the offline benchmark is at least $\Omega(1)$.
By contrast, the bad agents' utilities for the algorithm's allocation are low.
We will therefore let them compete for the same item in the Makeup Stage, limiting the algorithm's ability to catch up with the offline benchmark in terms of the bad agents' utilities. 

More precisely, for each good agent $a \in G_\ell$, let there be an item with supply $n^{s_\ell-1}$ and value $1$ only to agent $a$.
Let there be another item for the bad agents, where the item has supply $n^{s_L-1}$, and value $1$ only to the bad agents.

\paragraph{Choice of Group Sizes.}
Finally, we define the decreasing sequence $1 = s_0 > s_1 > \dots > s_L \ge 0$, and correspondingly, the group sizes.
Recall that $s_0 = 1$.
We define the rest of the sequence to satisfy a set of linear equations:
\begin{equation}
	\label{eq:lower-bound-neg-recurrence}
	s_{\ell-1} + |p| (1-s_\ell) = s_L(1+|p|) - \alpha
	~,
	\qquad \forall 1 \le \ell \le L~.
\end{equation}

The closed-form solution is:
\begin{equation}
	\label{eq:lower-bound-neg-closed-form}
	s_\ell = 1 - \frac{|p|-\alpha}{2\sum_{i=1}^L |p|^i +1} \sum_{i=L-\ell}^{L-1} |p|^i
	~,
	\qquad 
	\forall 0 \le \ell \le L
	~.
\end{equation}

By this closed-form, the sequence is indeed decreasing when $\alpha<|p|$.

\subsection{Properties of the Hard Instances}

We first show that the bad agents have low utilities for the online algorithm's allocation, because most of the items in the Upper Triangular Stage have been allocated to the good agents.

\begin{lemma}
    \label{lem:hardness-negative}
    The average utility of the agents in $B$ for the algorithm's allocation is less than $\nicefrac{(L+1)}{n}$.
\end{lemma}

\begin{proof}
    We will first show by induction that after the first $0 \le \ell \le L$ rounds of the Upper Triangular Stage, the agents outside groups $G_1$ to $G_{\ell}$ have average utility less than $\nicefrac{\ell}{n}$.

    The base case when $\ell=0$ is trivially true.
    
    Next, suppose the claim holds for some $\ell < L$, i.e., the agents outside groups $G_1$ to $G_\ell$ have average utilities less than $\nicefrac{\ell}{n}$ from the items in the first $\ell$ rounds of the Upper Triangular Stage. 
    Consider the item in round $\ell + 1$.
    Since there are $n^{s_\ell}$ ungroup agents and the item's supply is $n^{s_\ell-1} - n^{s_{\ell+1}-1}$, the ungrouped agents' average utility from this item is at most:
    \[
        \frac{n^{s_\ell-1} - n^{s_{\ell+1}-1}}{n^{s_\ell}} < \frac{1}{n}
        ~.
    \]
    
    Putting together with the induction hypothesis, the ungrouped agents' average utility after round $\ell+1$ is less than $\nicefrac{(\ell+1)}{n}$. 
    The claim for $\ell+1$  follows by our construction of group $G_{\ell+1}$:
    Since $G_{\ell+1}$ consists of a subset of these agents with the highest utilities, the remaining agents' average utility is at most the overall average, and thus, less than $\nicefrac{(\ell+1)}{n}$.
    
    As the special case when $\ell = L$, the bad agents in group $B$ have average utility less than $\nicefrac{L}{n}$ by the end of the Upper Triangular Stage.
    The lemma now follows since the Makeup Stage provides only $n^{s_L-1}$ additional utilities to these $n^{s_L}$ bad agents.
\end{proof}

Next, we upper bound the $p$-mean welfare of the online algorithm's allocation, using the fact that the average utility of the bad agents in group $B$ is low.

\begin{lemma}
	\label{lem:hardness-negative-algorithm}
	For hard instance $I(L, \alpha)$, we have:
	\[
		\ALG \le (L+1) \cdot n^{\frac{1-s_L-|p|}{|p|}}
		~.
	\]	
\end{lemma}

\begin{proof}
	By definition, we have:
	\begin{align*}
		\ALG & = \left( \frac{1}{n} \sum_{a \in A} \utility_a^p \right)^{\frac{1}{p}}\\
	    & \leq \left( \frac{1}{n} \sum_{a \in B} \utility_a^p \right)^{\frac{1}{p}} \tag{$B \subset A$ and $p < 0$} \\
	    & \le \left( \frac{1}{n} \cdot n^{s_L} \cdot \Big(\frac{L+1}{n}\Big)^p \right)^{\frac{1}{p}} \tag{\Cref{lem:hardness-negative} and Jensen's inequality} \\
		& = (L+1) \cdot n^{\frac{1-s_L-|p|}{|p|}}~.
	\end{align*}
\end{proof}

Finally, we lower bound the $p$-mean welfare of the offline optimal allocation.

\begin{lemma}
	\label{lem:hardness-negative-optimal}
	For  hard instance $I(L, \alpha)$, we have:
	
	\[
		\OPT \ge \left(1+\frac{L}{n^{\alpha}} \right)^{-\frac{1}{|p|}} n^{\frac{1-s_L(|p|+1)}{|p|}}
	\]	
\end{lemma}

\begin{proof}
The agents in group $B$ have value $1$ for all items in the Upper Triangular Stage and the last item in the Makeup Stage.
Distribute these items uniformly to the bad agents.
Then, the agents in group $B$ have utilities $n^{-s_L}$.

Further, allocate the every remaining item in the Makeup Stage to the unique agent who has value $1$ for it.
As a result, the agents in group $G_\ell$ for any $1 \le \ell \le L$ have utilities $n^{s_\ell-1}$. 

Hence, the optimal $p$-mean welfare is at least:
\begin{align*}
	\OPT & \geq \left(\frac{1}{n}\left( n^{s_L}\cdot \left( n^{-s_L}\right)^{p} +\sum_{\ell=1}^L \left( n^{s_{\ell-1}} -n^{s_\ell}\right)\cdot \left( n^{s_\ell -1}\right)^{p}\right)\right)^{\frac{1}{p}}\\
	& \geq \left(\frac{1}{n}\left( n^{s_L(|p|+1)} +\sum_{\ell=1}^L n^{s_{\ell-1}}\cdot \left( n^{s_\ell -1}\right)^{p}\right)\right)^{1/p} \tag{recall that $p < 0$} \\
	& =\left( n^{s_L(|p|+1) -1} +\sum_{\ell=1}^L n^{s_{\ell-1} +|p|( 1-s_\ell) -1}\right)^{1/p}\\
	& = \left(1+\frac{L}{n^{\alpha}} \right)^{-\frac{1}{|p|}} n^{\frac{1-s_L(|p|+1)}{|p|}}
	~.
	\tag{\Cref{eq:lower-bound-neg-recurrence}}
\end{align*}
\end{proof}

\begin{corollary}
	\label{cor:hardness-negative}
	For hard instance $I(L, \alpha)$, we have:
	\[
		\frac{\OPT}{\ALG} \ge \frac{\left(1+\frac{L}{n^{\alpha}}\right)^{-\frac{1}{|p|}} \cdot n^{1-s_L}}{L+1}
		~.
	\]
\end{corollary}

\subsection{Lower Bounds for Competitive Ratios}

In the range of $-O(\nicefrac{\log\log n}{\log n}) \le p \le - \nicefrac{1}{\log n}$, we resort to the $\Omega(\nicefrac{\log n}{\log \log n})$ lower bound by \Cref{thm:hardness-positive}.
Comparing this lower bound with the $O(n^{|p|}\log n)$ upper bound by Nashian Greedy (\Cref{thm:nashian}), we know that the optimal competitive ratio is poly-logarithmic in the number of agents, transitioning toward the polynomial dependence when $p$ is a negative constant.

We next present the lower bounds for $p \le - \omega(\nicefrac{\log\log n}{\log n})$.
\begin{theorem}
	\label{thm:hardness-negative-almost-nashian}
	For any $-o(1) \le	p \le - \omega(\nicefrac{\log\log n}{\log n})$ and any online algorithm for online $p$-mean welfare maximization, the competitive ratio is no smaller than:
    \[
    	n^{|p| - o(|p|)}
    	~,
    \]
    even when the agents have binary valuations (and unit monopolist utilities).
\end{theorem}

\begin{proof}
	Let $\alpha = \nicefrac{2\log \log n}{\log n} = o(|p|)$ and $L = 1$.
	By \Cref{eq:lower-bound-neg-closed-form}, we have:
	\begin{equation*}
	    1 - s_L = \frac{|p| - \alpha}{2|p| +1} = |p| - o(|p|).
	\end{equation*}
	
	Further:
	\[
		\left(1+\frac{1}{n^{\alpha}}\right)^{-\frac{1}{|p|}} \ge e^{-\frac{1}{|p| n^\alpha}} = n^{-\frac{1}{|p| \log^2 n}} = n^{-o(|p|)}
		~. 
	\]
	
	Therefore, by \Cref{cor:hardness-negative}:
	\begin{equation*}
	    \frac{\OPT}{\ALG} \ge \frac{\left(1+\frac{1}{n^{\alpha}}\right)^{-\frac{1}{|p|}} \cdot n^{1-s_L}}{2} = n^{|p|-o(|p|)}.
	\end{equation*}	
\end{proof}

\begin{theorem}
	\label{thm:hardness-nashian-to-harmonic}
	For any $-1 \le p \le \Omega(1)$ and any online algorithm for online $p$-mean welfare maximization, the competitive ratio is no smaller than:
	\[
		n^{\frac{|p|}{|p|+1}-o(1)}
        ,
	\]
    even when the agents have binary valuations (and unit monopolist utilities).
\end{theorem}

\begin{proof}
	Let $\alpha = 0$ and $L = \log n$.
	By \Cref{eq:lower-bound-neg-closed-form}, we have:
	\begin{align*}
	    1 - s_L 
	    &
	    = \frac{\sum_{i=1}^L |p|^i}{2\sum_{i=1}^L |p|^i +1}  \\
	    &
	    = \frac{|p|}{|p|+1} -\frac{|p|^{L+1}}{(2\sum_{i=1}^L |p|^i+1)(|p|+1)} \\
	    &
	    > \frac{|p|}{|p|+1} - \frac{1}{\log n}
	    ~.
	    \tag{$|p|^i \ge |p|^{L+1}$ and $L = \log n$}
	\end{align*}
	
	Therefore, for our choice of $\alpha$ and $L$, \Cref{cor:hardness-negative} gives:
	\begin{equation*}
	    \frac{\OPT}{\ALG} \geq n^{\frac{|p|}{|p|+1} - \frac{1}{\log n}} \cdot \left( 1+ \log n \right)^{-(1+\frac{1}{|p|})} \geq n^{\frac{|p|}{|p|+1}(1-o(1))}.
	\end{equation*}	
\end{proof}

\begin{theorem}
	\label{thm:hardness-harmonic-to-egalitarian}
	For any $-\infty \le p \le -1$ and any online algorithm for online $p$-mean welfare maximization, the competitive ratio is no smaller than:
	\[
		n^{\frac{1}{2}-o(1)}
        .
	\]
\end{theorem}

\begin{proof}
	Let $\alpha = 0$ and $L = \log n$, we have:
	\begin{align*}
	    1 - s_L 
	    &
	    = \frac{\sum_{i=1}^L |p|^i}{2\sum_{i=1}^L |p|^i +1} \\
	    &
	    = \frac{1}{2} -\frac{1}{4\sum_{i=1}^L |p|^i+2} \\ 
	    &
	    \geq \frac{1}{2} - \frac{1}{4\log n+2}
	    ~.
	    \tag{$|p| \ge 1$ and $L = \log n$}
	\end{align*}
	
	Therefore, from \Cref{cor:hardness-negative} we derive the bound:
	\begin{equation*}
	    \frac{\OPT}{\ALG} \geq n^{\frac{1}{2}-\frac{1}{4\log n+2}} \cdot (1+\log n)^{-(1+\frac{1}{|p|})} = n^{\frac{1}{2}-o(1)}
	    .
	\end{equation*}	
\end{proof}

\bibliographystyle{abbrvnat}
\bibliography{reference}

\begin{thebibliography}{34}
\providecommand{\natexlab}[1]{#1}
\providecommand{\url}[1]{\texttt{#1}}
\expandafter\ifx\csname urlstyle\endcsname\relax
  \providecommand{\doi}[1]{doi: #1}\else
  \providecommand{\doi}{doi: \begingroup \urlstyle{rm}\Url}\fi

\bibitem[Aleksandrov and Walsh(2020)]{AleksandrovW:AAAI:2020}
M.~Aleksandrov and T.~Walsh.
\newblock Online fair division: A survey.
\newblock In \emph{Proceedings of the 34th AAAI Conference on Artificial Intelligence}, pages 13557--13562, 2020.

\bibitem[Aleksandrov et~al.(2015)Aleksandrov, Aziz, Gaspers, and Walsh]{AleksandrovAGW:IJCAI:2015}
M.~Aleksandrov, H.~Aziz, S.~Gaspers, and T.~Walsh.
\newblock Online fair division: analysing a food bank problem.
\newblock In \emph{Proceedings of the 24th International Joint Conference on Artificial Intelligence}, pages 2540--2546, 2015.

\bibitem[Alon et~al.(2006)Alon, Awerbuch, Azar, Buchbinder, and Naor]{AlonAABN:TALG:2006}
N.~Alon, B.~Awerbuch, Y.~Azar, N.~Buchbinder, and J.~Naor.
\newblock A general approach to online network optimization problems.
\newblock \emph{ACM Transactions on Algorithms}, 2\penalty0 (4):\penalty0 640--660, 2006.

\bibitem[Banerjee et~al.(2022)Banerjee, Gkatzelis, Gorokh, and Jin]{BanerjeeGGJ:SODA:2022}
S.~Banerjee, V.~Gkatzelis, A.~Gorokh, and B.~Jin.
\newblock Online {N}ash social welfare maximization with predictions.
\newblock In \emph{Proceedings of the 33rd Annual ACM-SIAM Symposium on Discrete Algorithms}, pages 1--19. SIAM, 2022.

\bibitem[Banerjee et~al.(2023)Banerjee, Gkatzelis, Hossain, Jin, Micha, and Shah]{BanerjeeGHJMS:IJCAI:2023}
S.~Banerjee, V.~Gkatzelis, S.~Hossain, B.~Jin, E.~Micha, and N.~Shah.
\newblock Proportionally fair online allocation of public goods with predictions.
\newblock In \emph{Proceedings of the 32nd International Joint Conference on Artificial Intelligence}, pages 20--28, 2023.

\bibitem[Bansal and Sviridenko(2006)]{BansalS:STOC:2006}
N.~Bansal and M.~Sviridenko.
\newblock The santa claus problem.
\newblock In \emph{Proceedings of the 38th Annual ACM Symposium on Theory of Computing}, page 31–40, 2006.

\bibitem[Barman et~al.(2018)Barman, Krishnamurthy, and Vaish]{BarmanKV:EC:2018}
S.~Barman, S.~K. Krishnamurthy, and R.~Vaish.
\newblock Finding fair and efficient allocations.
\newblock In \emph{Proceedings of the 19th ACM Conference on Economics and Computation}, pages 557--574, 2018.

\bibitem[Barman et~al.(2020)Barman, Bhaskar, Krishna, and Sundaram]{BarmanBKS:ESA:2020}
S.~Barman, U.~Bhaskar, A.~Krishna, and R.~G. Sundaram.
\newblock Tight approximation algorithms for p-mean welfare under subadditive valuations.
\newblock In \emph{Proceedings of the 28th Annual European Symposium on Algorithms}, pages 11:1--11:17, 2020.

\bibitem[Barman et~al.(2022)Barman, Khan, and Maiti]{BarmanKM:AAAI:2022}
S.~Barman, A.~Khan, and A.~Maiti.
\newblock Universal and tight online algorithms for generalized-mean welfare.
\newblock In \emph{Proceedings of the 36th AAAI Conference on Artificial Intelligence}, pages 4793--4800, 2022.

\bibitem[Benad\`{e} et~al.(2018)Benad\`{e}, Kazachkov, Procaccia, and Psomas]{BenadeKPP:EC:2018}
G.~Benad\`{e}, A.~M. Kazachkov, A.~D. Procaccia, and C.-A. Psomas.
\newblock How to make envy vanish over time.
\newblock In \emph{Proceedings of the 19th ACM Conference on Economics and Computation}, pages 593--610, 2018.

\bibitem[Bez{\'a}kov{\'a} and Dani(2005)]{BezakovaD:SIGecom:2005}
I.~Bez{\'a}kov{\'a} and V.~Dani.
\newblock Allocating indivisible goods.
\newblock \emph{ACM SIGecom Exchanges}, 5\penalty0 (3):\penalty0 11--18, 2005.

\bibitem[Buchbinder and Naor(2009)]{BuchbinderN:MOR:2009}
N.~Buchbinder and J.~Naor.
\newblock Online primal-dual algorithms for covering and packing.
\newblock \emph{Mathematics of Operations Research}, 34\penalty0 (2):\penalty0 270--286, 2009.

\bibitem[Caragiannis et~al.(2019)Caragiannis, Kurokawa, Moulin, Procaccia, Shah, and Wang]{CaragiannisKMPSW:TEAC:2019}
I.~Caragiannis, D.~Kurokawa, H.~Moulin, A.~Procaccia, N.~Shah, and J.~Wang.
\newblock The unreasonable fairness of maximum nash welfare.
\newblock \emph{ACM Transactions on Economics and Computation}, 7\penalty0 (3):\penalty0 1--32, 2019.

\bibitem[Chakrabarty et~al.(2009)Chakrabarty, Chuzhoy, and Khanna]{ChakrabartyCK:FOCS:2009}
D.~Chakrabarty, J.~Chuzhoy, and S.~Khanna.
\newblock On allocating goods to maximize fairness.
\newblock In \emph{Proceedings of the 50th Annual IEEE Symposium on Foundations of Computer Science}, pages 107--116, 2009.

\bibitem[Chaudhury et~al.(2021)Chaudhury, Garg, and Mehta]{ChaudhuryGM:AAAI:2021}
B.~R. Chaudhury, J.~Garg, and R.~Mehta.
\newblock Fair and efficient allocations under subadditive valuations.
\newblock In \emph{Proceedings of the 35th AAAI Conference on Artificial Intelligence}, pages 5269--5276, 2021.

\bibitem[Cohen and Panigrahi(2023)]{CohenP:ICALP:2023}
I.~R. Cohen and D.~Panigrahi.
\newblock A general framework for learning-augmented online allocation.
\newblock In \emph{Proceedings of the 50th International Colloquium on Automata, Languages, and Programming}, pages 43:1--43:21, 2023.

\bibitem[Devanur and Jain(2012)]{DevanurJ:STOC:2012}
N.~R. Devanur and K.~Jain.
\newblock Online matching with concave returns.
\newblock In \emph{Proceedings of the 44th Annual ACM Symposium on Theory of Computing}, pages 137--144, 2012.

\bibitem[Dobzinski et~al.(2024)Dobzinski, Li, Rubinstein, and Vondr{\'a}k]{DobzinskiLRV:STOC:2024}
S.~Dobzinski, W.~Li, A.~Rubinstein, and J.~Vondr{\'a}k.
\newblock A constant-factor approximation for nash social welfare with subadditive valuations.
\newblock In \emph{Proceedings of the 56th Annual ACM Symposium on Theory of Computing}, pages 467--478, 2024.

\bibitem[Feng and Li(2024)]{FengL:ICALP:2024}
Y.~Feng and S.~Li.
\newblock A note on approximating weighted nash social welfare with additive valuations.
\newblock In \emph{Proceedings of the 51st EATCS International Colloquium on Automata, Languages and Programming}, 2024.

\bibitem[Garg et~al.(2020)Garg, Kulkarni, and Kulkarni]{GargKK:SODA:2020}
J.~Garg, P.~Kulkarni, and R.~Kulkarni.
\newblock Approximating {N}ash social welfare under submodular valuations through (un)matchings.
\newblock In \emph{Proceedings of the 31st Annual ACM-SIAM Symposium on Discrete Algorithms}, pages 2673--2687, 2020.

\bibitem[Ghodsi et~al.(2011)Ghodsi, Zaharia, Hindman, Konwinski, Shenker, and Stoica]{GhodsiZHKSS:NSDI:2011}
A.~Ghodsi, M.~Zaharia, B.~Hindman, A.~Konwinski, S.~Shenker, and I.~Stoica.
\newblock Dominant resource fairness: Fair allocation of multiple resource types.
\newblock In \emph{Proceedings of the 8th USENIX Symposium on Networked Systems Design and Implementation}, 2011.

\bibitem[Gkatzelis et~al.(2021)Gkatzelis, Psomas, and Tan]{GkatzelisPT:AAAI:2021}
V.~Gkatzelis, A.~Psomas, and X.~Tan.
\newblock Fair and efficient online allocations with normalized valuations.
\newblock In \emph{Proceedings of the 35th AAAI Conference on Artificial Intelligence}, pages 5440--5447, 2021.

\bibitem[Hajiaghayi et~al.(2022)Hajiaghayi, Panigrahi, Khani, and Springer]{HajiaghayiPKS:NeurIPS:2022}
M.~Hajiaghayi, D.~Panigrahi, M.~Khani, and M.~Springer.
\newblock Online algorithms for the {S}anta {C}laus problem.
\newblock In \emph{Proceedings of the 36th Annual Conference on Advances in Neural Information Processing Systems}, pages 30732--30743, 2022.

\bibitem[Huang et~al.(2023)Huang, Li, Shu, and Wei]{HuangLSW:WINE:2023}
Z.~Huang, M.~Li, X.~Shu, and T.~Wei.
\newblock Online nash welfare maximization without predictions.
\newblock In \emph{Proceedings of the 19th International Conference on Web and Internet Economics}, pages 402--419. Springer, 2023.

\bibitem[Huang et~al.(2024)Huang, Tang, and Wajc]{HuangTW:SIGecom:2024}
Z.~Huang, Z.~G. Tang, and D.~Wajc.
\newblock Online matching: A brief survey.
\newblock \emph{ACM SIGecom Exchanges}, 22:\penalty0 135--158, 2024.

\bibitem[Kalyanasundaram and Pruhs(2000)]{KalyanasundaramP:TCS:2000}
B.~Kalyanasundaram and K.~R. Pruhs.
\newblock An optimal deterministic algorithm for online b-matching.
\newblock \emph{Theoretical Computer Science}, 233\penalty0 (1-2):\penalty0 319--325, 2000.

\bibitem[Karp et~al.(1990)Karp, Vazirani, and Vazirani]{KarpVV:STOC:1990}
R.~Karp, U.~Vazirani, and V.~Vazirani.
\newblock An optimal algorithm for on-line bipartite matching.
\newblock In \emph{Proceedings of the 22nd Annual ACM Symposium on Theory of Computing}, pages 352--358, 1990.

\bibitem[Lee(2017)]{Lee:IPL:2017}
E.~Lee.
\newblock {APX}-hardness of maximizing {N}ash social welfare with indivisible items.
\newblock \emph{Information Processing Letters}, 122:\penalty0 17--20, 2017.

\bibitem[Li and Vondr{\'a}k(2022)]{LiV:FOCS:2022}
W.~Li and J.~Vondr{\'a}k.
\newblock A constant-factor approximation algorithm for nash social welfare with submodular valuations.
\newblock In \emph{Proceedings of the 62nd Annual IEEE Symposium on Foundations of Computer Science}, pages 25--36, 2022.

\bibitem[Mehta(2013)]{Mehta:FTTCS:2013}
A.~Mehta.
\newblock Online matching and ad allocation.
\newblock \emph{Foundations and Trends in Theoretical Computer Science}, 8\penalty0 (4):\penalty0 265--368, 2013.

\bibitem[Moulin(2004)]{Moulin:2004}
H.~Moulin.
\newblock \emph{Fair division and collective welfare}.
\newblock MIT press, 2004.

\bibitem[Varian(1974)]{Varian:1974}
H.~R. Varian.
\newblock Equity, envy, and efficiency.
\newblock \emph{Journal of Economic Theory}, 9\penalty0 (1):\penalty0 63--91, 1974.

\bibitem[Walsh(2011)]{Walsh:ADT:2011}
T.~Walsh.
\newblock Online cake cutting.
\newblock In \emph{Proceedings of the 2nd International Conference on Algorithmic Decision Theory}, pages 292--305. Springer, 2011.

\bibitem[Zeng and Psomas(2020)]{ZengP:EC:2020}
D.~Zeng and A.~Psomas.
\newblock Fairness-efficiency tradeoffs in dynamic fair division.
\newblock In \emph{Proceedings of the 21st ACM Conference on Economics and Computation}, pages 911--912, 2020.

\end{thebibliography}

\appendix

\section{Sufficiently Positive Regime: $p \ge \nicefrac{1}{\log n}$}
\label{app:positive-regime}
	
\subsection{Online Primal Dual Framework}

In the positive regime, we follow the online primal dual framework under Fenchel duality.
First, we present an equivalent primal program of the offline optimal benchmark for any $p\in (0,1]$, where we transform the $p$-mean welfare by raising to the $p$-th power:
\begin{equation}
\tag{\text{\rm\bf Primal}}
\label{eq:p>0-primal}
\begin{aligned}
	\text{maximize} \quad & P ~\defeq~ \frac{1}{p} \sum_{a\in A} \utility_a^p \\
	\text{subject to} \quad & \sum_{a \in A} x_a(t) \le 1 && \forall t \in I \\
	& \utility_a = \int_0^T x_a(t) v_a(t) \dif{t} && \forall a \in A \\[2ex]
	& x_a(t) \ge 0 && \forall a \in A, ~\forall t \in I
\end{aligned}
\end{equation}

Consequently, an online algorithm is $\Gamma$-competitive for online $p$-mean welfare maximization if and only if it is $\Gamma^p$-competitive for \ref{eq:p>0-primal}.

The corresponding Fenchel dual convex program of \ref{eq:p>0-primal} is:
\begin{align*}
	\text{minimize} \quad & \int_0^T \alpha(t) \dif{t} + \left( \frac{1}{p} - 1 \right) \sum_{a\in A} \beta_a^{-\frac{p}{1-p}} \\
	\text{subject to} \quad & \alpha(t) \geq v_{a}(t)\beta_a && \forall a \in A, ~\forall t \in I \\[2ex]
	& \beta_a \geq 0 && \forall a \in A, ~\forall t \in I
\end{align*}

It is more convenient in the subsequent analysis to change variables with $\gamma_a = {\beta_a}^{- \nicefrac{1}{(1-p)}}$, to make the (online) relation between $\utility_a$ and $\gamma_a$ ``almost homogeneous''.
Hence, we rewrite the dual program as follows:
\begin{equation}
\tag{\text{\rm\bf Dual}}
\label{eq:p>0-dual}
\begin{aligned}
	\text{minimize} \quad & 
	D ~\defeq~ \int_0^T \alpha(t) \dif{t} + \left( \frac{1}{p} - 1 \right) \sum_{a\in A} \gamma_a^{p}  \\
	\text{subject to} \quad & \alpha(t) \geq v_{a}(t) \gamma_a^{-(1-p)} && \forall a \in A, ~\forall t \in I \\[2ex]
	& \gamma_a \geq 0 && \forall a \in A, ~\forall t \in I
\end{aligned}
\end{equation}

An online primal dual algorithm will maintain not only an online allocation $x$ and the agents' corresponding utilities $\utility_a$'s as the primal assignment, but also a dual assignment $\alpha$ and $\gamma$.
The next lemma summarizes the sufficient conditions under which an online primal-dual algorithm is $\Gamma$-competitive.
Recall that $P$ and $D$ denote the primal and dual objective values.

\begin{lemma}
	\label{lem:p>0-primal-dual}
	Suppose that an online primal dual algorithm $\A$ satisfies the following conditions:
	\begin{enumerate}[label=(\alph*)]
		\item $x = \big( x_a(t) \big)_{a \in A, t \in I}$ forms a feasible assignment of \ref{eq:p>0-primal}; \label{cond:p>0-primal-feasibility}
		\item $\alpha = \big( \alpha(t) \big)_{t \in I}$ and $\gamma = \big( \gamma_a \big)_{a \in A}$ form a feasible assignment of \ref{eq:p>0-dual}; and \label{cond:p>0-dual-feasibility}
		\item $\Gamma^p \cdot P \geq D$~. \label{cond:p>0-ratio}
	\end{enumerate}
	Then, the algorithm is $\Gamma$-competitive for online $p$-mean welfare maximization.
\end{lemma}

\begin{proof}
	First, condition~\ref{cond:p>0-primal-feasibility} ensures that algorithm $\A$ makes a feasible allocation.
	
	Next, we prove the competitive ratio.
	Recall that $\OPT$ denotes the optimal $p$-mean welfare.
	Denote $P^*$ and $D^*$ as the optimal primal and dual objective values.
	By the relation between \ref{eq:p>0-primal} and the $p$-mean welfare maximization problem, we have:
	\[
		\OPT ~=~ \Big(\frac{p}{n} \cdot P^*\Big)^{\nicefrac{1}{p}} 
		~.
	\]
	
	By weak duality, we have $P^* \le D^*$. 
	By condition~\ref{cond:p>0-dual-feasibility}, the dual assignment is feasible, which implies $D^* \leq D$ because the dual is a minimization problem.
    Hence, we get that:
    \[
    	\OPT 
        ~\leq~ 
        \Big(\frac{p}{n} \cdot D\Big)^{\nicefrac{1}{p}} 
        ~.
    \]
    
    Combining with condition~\ref{cond:p>0-ratio} that $D \le \Gamma^p \cdot P$, we conclude that:
	\begin{equation*}
        \OPT 
        ~\leq~ 
        \Gamma \cdot \Big(\frac{p}{n} \cdot P(x)\Big)^{\nicefrac{1}{p}}
        ~.
	\end{equation*}
	
	Finally, by the relation between \ref{eq:p>0-primal} and the $p$-mean welfare maximization problem, the right hand side equals $\Gamma \cdot \ALG$.
	In other words, the algorithm is $\Gamma$-competitive.
\end{proof}
	
The design and analysis of online primal dual algorithms are driven by the primal and dual programs.
The Karush-Kuhn-Tucker (KKT) conditions assert that the optimal primal solution $x^*$, $\utility^*$ and the optimal dual solution $\alpha^*$, $\beta^*$, and the corresponding $\gamma^*$ satisfy that:
\[
	\beta_a^* = \big( \utility_a^* \big)^{-(1-p)}
\]
and, thus, by $\gamma_a^* = (\beta_a^*)^{-\nicefrac{1}{(1-p)}}$ we have:
\[
	\gamma_a^* = \utility_a^*
	~.
\]

Further, if an item $t$ is (partially) allocated to an agent $a$, i.e., if $x_a(t) > 0$, then the corresponding dual constraint is tight:
\[
	\alpha^*(t) = v_a(t) \big(\gamma_a^*\big)^{-(1-p)} = v_a(t) \big(\utility_a^*\big)^{-(1-p)}
	~.
\]

Since the partial derivative of $U_a$ in $x_a(t)$ equals $v_a(t) U_a^{-(1-p)}$, the above condition says that every item is allocated to the agents that yield the maximum marginal increase in the objective of \ref{eq:p>0-primal}, conditioned on the allocation of the other items.	

When an online algorithm decides how to allocate item $t$, it only knows an agent $a$'s current utility $\utility_a$ but not the final one.
Hence, the algorithm has insufficient information to decide which agent yields the maximum increase in the objective with respect to the final utility.
The online primal dual algorithms in this section will project the agents' final utilities based on the current ones, and choose the primal and dual assignments based on the projected utilities.

\subsection{Greedy}

\ref{eq:p>0-primal} as an online algorithm problem is a special case of online matching with concave returns studied by \citet{DevanurJ:STOC:2012}.
Applying their results to degree-$p$ polynomial functions, we have a $p^{-p}$ competitive ratio for \ref{eq:p>0-primal}, and correspondingly a $\nicefrac{1}{p}$ competitive ratio for $p$-mean welfare maximization.
We summarize this result below as a theorem, and present their algorithm and analysis to be self-contained.
         
\begin{theorem}[c.f.\ \citet{DevanurJ:STOC:2012}]
	\label{thm:p>0-1/p}
	For any $p\in (0,1]$, there is a $\nicefrac{1}{p}$-competitive primal dual algorithm for online $p$-mean welfare maximization.
\end{theorem}

\paragraph{Greedy as an Online Primal Dual Algorithm.}
Na\"ively, we may project the final utility to be $C$ times larger than the current utility for some factor $C > 1$, and assign values to the primal and dual variables based on the projected utilities.
The online primal dual algorithm of \citet{DevanurJ:STOC:2012} chooses $C = \nicefrac{1}{p}$.
Hence, we let $\gamma_a = \nicefrac{\utility_a}{p}$, $\alpha(t) = \max_{a \in A} v_a(t) \gamma_a^{-(1-p)}$, and allocate each item $t$ to an agent that achieves this maximum.

Interestingly, the projected utilities of the agents merely scale their current utilities by the same factor.
Hence, the primal allocation rules with and without the scaling are identical.
In other words, the algorithm of \citet{DevanurJ:STOC:2012} is identical to the greedy algorithm that allocates each item to maximize the immediate increase of the $p$-mean welfare;
the projected utilities mainly play their parts in the dual assignment and the analysis.

\begin{algorithm}{Greedy (as an Online Primal Dual Algorithm)}%
    \emph{Initialization:~}
   	Let $x_a(t) = \utility_a = \alpha(t) = \gamma_a = 0$ for any agent $a\in A$ and any time $t\in I$.\\[2ex]
	\emph{Invariants:~}
	Maintain at all time that for any agent $a\in A$:
    \begin{equation}
    	\label{eq:invariant-devanur-jain}
    	\utility_a = \int_0^T v_a(t) x_a(t) \dif{t} \quad,\quad \gamma_a ~=~ \gamma(\utility_a) ~\defeq~ \frac{\utility_a}{p}
		~.
	\end{equation}

	\emph{Online Decisions:~}
	For each item $t\in I$:
	\begin{enumerate}
		\item Let $\alpha(t) = \max_{a \in A} v_a(t) \gamma_a^{-(1-p)}$.
		\item Allocate the item to an agent that achieves the above maximum.
	\end{enumerate}
\end{algorithm}
	
\medskip

\begin{proof}[Proof of Theorem~\ref{thm:p>0-1/p}]
	We will prove the theorem by verifying that Greedy as a primal dual algorithm satisfies conditions \ref{cond:p>0-primal-feasibility}-\ref{cond:p>0-ratio} in \Cref{lem:p>0-primal-dual} with $\Gamma= \nicefrac{1}{p}$.
	
	Condition \ref{cond:p>0-primal-feasibility} holds because the algorithm allocates each item to one agent and maintains the first invariant in \Cref{eq:invariant-devanur-jain}.
	
	Next, consider condition \ref{cond:p>0-dual-feasibility}. 
	By definition, we set $\alpha(t) = \max_{a \in A} v_a(t) \gamma_a^{-(1-p)}$ \emph{at the time when we allocate item $t$}.
	Hence, for any agent $a \in A$ we have the desired dual constraint:
	\[
		\alpha(t) \geq v_a(t)\gamma_a^{-(1-p)}
		~,
	\]
	for the values of $\gamma_a$'s at that time.
	It remains to verify that $\gamma_a$ can only become larger at the end.
	By the second invariant in \Cref{eq:invariant-devanur-jain}, $\gamma_a$ is non-negative and increasing in $\utility_a$.
	Further, an agent $a$'s utility $\utility_a$ is non-decreasing over time.
	We conclude that $\gamma_a$ is non-decreasing over time.
		
	Finally, we turn to condition \ref{cond:p>0-ratio}. 
	By the initial primal and dual assignments, we have that $P(x) = D(\alpha, \gamma) = 0$ at the beginning, which satisfies the condition.
	Next, consider any time $t \in I$.
	By definition, the item is allocated to an agent $a$ with the maximum $v_a(t) \gamma_a^{-(1-p)}$, which is equivalent to having the maximum $v_a(t) \utility_a^{-(1-p)}$.
	
	As a result of the allocation at time $t$, the primal changes at rate:
	\begin{equation*}
		\frac{\dif{P}}{\dif{t}} ~=~ \underbrace{\vphantom{\Big|} v_a(t)}_{\nicefrac{\dif{\utility_a}}{\dif{t}}} \,\cdot~ \underbrace{\vphantom{\Big|} \utility_a^{-(1-p)}}_{\nicefrac{\dif{P}}{\dif{\utility_a}}}
		~.
	\end{equation*}

	The dual changes at rate:
	\[
		\frac{\dif{D}}{\dif{t}}
		~=~ 
		\underbrace{\vphantom{\bigg|} v_a(t) \gamma_a^{-(1-p)}}_{\alpha(t)} 
		~+~ 
		\frac{1-p}{p} 
		\,\cdot\, 
		\underbrace{\vphantom{\bigg|} p \cdot \gamma_a^{-(1-p)}}_{\nicefrac{\dif{D}}{\dif{\gamma_a}}} 
		\,\cdot\,
		\underbrace{\vphantom{\bigg|} \frac{1}{p}}_{\nicefrac{\dif{\gamma_a}}{\dif{\utility_a}}} 
		\,\cdot\,
		\underbrace{\vphantom{\bigg|} v_a(t)}_{\nicefrac{\dif{\utility_a}}{\dif{t}}}
		~=~
		\frac{1}{p} v_a(t) \gamma_a^{-(1-p)}
		~.		
	\]
	
	By $\gamma_a = \nicefrac{\utility_a}{p}$, we get that:
	\[
		\frac{\dif{P}}{\dif{t}} ~=~ p^p \cdot \frac{\dif{D}}{\dif{t}}
		~.
	\]
	
	Hence, condition~\ref{cond:p>0-ratio} holds with equality for $\Gamma = \nicefrac{1}{p}$ at all time.
\end{proof}

\section{Online Primal Dual in the Positive Nashian Regime}
\label{app:positive-nashian}

Our result is an $\Theta(\log n)$-competitive primal dual algorithm for any $p\in [0,1]$.
Recall that it suffices to prove this competitive ratio for the relaxed problem in which agents start with $\nicefrac{1}{n}$ base utilities.
We modify the online primal dual algorithm of \citet{DevanurJ:STOC:2012} by designing a suitable regularized utility function.

\begin{algorithm}{Regularized Greedy (as an Online Primal Dual Algorithm)}%
	\emph{Initialization:~}
	For any agent $a\in A$ and any time $t\in I$, let:
	\[
		x_a(t) = 0 ~,\quad \utility_a = \frac{1}{n} ~,\quad \alpha(t) = 0 ~,\quad \gamma_a = \frac{\log(n+1)+1}{n}
		~.
	\]
	\emph{Invariants:~}
	Maintain at all time that for any agent $a\in A$:
	\begin{equation}
		\label{eq:invariant-log(n)}
		\utility_a = \int_0^T v_a(t) x_a(t) \dif{t} \quad,\quad \gamma_a = \gamma(\utility_a) \defeq \utility_a \cdot \bigg(1+\log\Big(1+\frac{1}{n}\Big)-\log \utility_a\bigg)
		~.
	\end{equation}
	
	\emph{Online Decisions:~}
	For each item $t\in I$:
	\begin{enumerate}
		\item Allocate the item to an agent $a$ that with the maximum $v_a(t) \cdot \gamma_a^{-(1-p)}$.	
		\item Let $\alpha(t) = v_a(t) \cdot \gamma_a^{-(1-p)}$.
	\end{enumerate}
\end{algorithm}

\begin{lemma}
    \label{lem:p>0-invariant-gamma_a-properties}
    The regularized utility $\gamma(\utility)$ is increasing in $\utility \in [\nicefrac{1}{n}, 1+\nicefrac{1}{n}]$.
\end{lemma}

\begin{theorem}
    \label{thm:p>0-log(n)}
    For any $0 < p \le \nicefrac{1}{\log n}$, Regularized Greedy is $O(\log n)$-competitive for online $p$-mean welfare maximization.
\end{theorem}

\begin{proof}%
	We will prove the theorem by verifying that Regularized Greedy as a primal dual algorithm satisfy conditions \ref{cond:p>0-primal-feasibility}-\ref{cond:p>0-ratio} in \Cref{lem:p>0-primal-dual} with:
	\[
		\Gamma=\log(n+1)+1
		~.
	\]
	
	Condition \ref{cond:p>0-primal-feasibility} holds because the algorithm allocates each item to one agent and maintains the first invariant in \Cref{eq:invariant-log(n)}.

     Next, consider condition \ref{cond:p>0-dual-feasibility}. 
	By definition, we set $\alpha(t) = \max_{a \in A} v_a(t) \gamma_a^{-(1-p)}$ \emph{at the time when we allocate item $t$}.
	Hence, for any agent $a \in A$ we have the desired dual constraint:
	\[
		\alpha(t) \geq v_a(t)\gamma_a^{-(1-p)}
		~,
	\]
	for the values of $\gamma_a$'s at that time.
	Further, $\gamma_a$ is non-decreasing over time because $U_a$ is non-decreasing over time, and $\gamma_a$ is non-decreasing in $U_a$ (first property of \Cref{lem:p>0-invariant-gamma_a-properties}).
	Hence, the dual constraint also holds for the final value of $\gamma_a$.

	Finally, we turn to condition \ref{cond:p>0-ratio}. 
	By the initial primal and dual assignments, we have that:
	\[
		P = n\cdot \frac1p \cdot \left({\frac{1}{n}}\right)^p
		~,
	\]
	and:
	\[
		D = n\cdot \left(\frac1p -1\right) \cdot
	\left(\frac{\log(n+1)+1}{n}\right)^p = n\cdot \left(\frac1p -1\right) \cdot
	\left(\frac{\Gamma}{n}\right)^p
	\]
 	at the beginning, which satisfies the condition.
 	
	Next, consider any time $t \in I$.
	Suppose that the algorithm allocates item $t$ to agents $a$.
	The primal changes at rate:
	\begin{equation*}
		\frac{\dif{P}}{\dif{t}} ~=~ \underbrace{\vphantom{\Big|} v_a(t)}_{\nicefrac{\dif{\utility_a}}{\dif{t}}} \,\cdot~ \underbrace{\vphantom{\Big|} \utility_a^{-(1-p)}}_{\nicefrac{\dif{P}}{\dif{\utility_a}}}
		~.
	\end{equation*}

	The dual changes at rate:
	\begin{align*}
	\frac{\dif{D}}{\dif{t}}
	& ~=~ 
	\underbrace{\vphantom{\bigg|} v_a \gamma_a^{-(1-p)}}_{\alpha(t)}
	~+~ 
	\frac{1-p}{p} 
	\,\cdot\, 
	\underbrace{\vphantom{\bigg|} p \cdot \gamma_a^{-(1-p)}}_{\nicefrac{\dif{D}}{\dif{\gamma_a}}} 
	\,\cdot\,
	\frac{\dif{\gamma_a}}{\dif{\utility_a}}
	\,\cdot\,
	\underbrace{\vphantom{\bigg|} v_a(t)}_{\nicefrac{\dif{\utility_a}}{\dif{t}}} ~.
	\end{align*}
	
	By multiplying both sides of $\Gamma^p \cdot \frac{\dif{P}}{\dif{t}} \geq \frac{\dif{D}}{\dif{t}}$ by $\frac{\gamma_a^{1-p}}{v_a(t)}$, we get that:
	\begin{equation*}
		\Gamma^p \cdot \left(\frac{\gamma_a}{\utility_a}\right)^{1-p}
		~\geq~
		1 + (1-p) \frac{\dif{\gamma_a}}{\dif{\utility_a}}~.
	\end{equation*}
	
	Define an auxiliary variable:
	\[
		z_a = z (\utility_a) \defeq \dfrac{\gamma (\utility_a)}{\utility_a} = 1 + \log\Big(1+\frac{1}{n}\Big)-\log \utility_a
		~.
	\]
	
	We have:
	\[
		\frac{\dif{\gamma_a}}{\dif{\utility_a}}
    	= \log\Big(1+\frac{1}{n}\Big)-\log \utility_a
    	= z_a - 1
    	~.
	\]
	
	Hence, the condition can be equivalently stated as:
	\begin{equation}
		\label{eq:p>0-ratio-desired-ineq-transformed-log(n)}
		\Gamma ^p \cdot z_a^{1-p}
		~\geq~ 1 + (1-p) (z_a - 1)~.
	\end{equation}
	
	This holds for $z_a = 1$ because $\Gamma > 1$. 

	It remains to verify that the derivative of the left-hand-side is greater than or equal to that of the right-hand-side, i.e.:
	\begin{equation*}
		(1-p) \bigg(\frac{\Gamma}{z_a}\bigg)^p \ge 1-p
		~.
	\end{equation*}
	
	This holds because $U_a \ge \nicefrac{1}{n}$, which implies $z_a \le \log (n+1)+1 = \Gamma$.
\end{proof}

\section{Generalizing to Non-Unit Monopolist Utilities}
\label{sec:discussion}

This paper introduces algorithms designed to maximize $p$-mean welfare for various ranges of $p$. To facilitate a clearer comparison with previous work (e.g., \citet{BarmanKM:AAAI:2022}), we initially focused on a simplified model where all agents have unit monopolist utilities. This approach allowed us to isolate the complexities introduced by non-unit monopolist utilities. 

In this section, we extend our discussion to generalize these algorithms for scenarios where agents possess non-unit monopolist utilities. Assume that each agent $a$'s monopolist utility is accurately predicted to be $V_a\geq 0$, and the max-min ratio of the monopolist utilities is a constant, i.e.,
\[K \defeq \frac{\max_{a\in A} V_a}{\min_{a\in A} V_a}=O(1).\]

Our results in this general setting are summarized as follows:

\begin{enumerate}
	\item For $\nicefrac{1}{\log n}\le p\le 1$, the online primal dual algorithm by \citet{DevanurJ:STOC:2012}, remains $\nicefrac{1}{p}$-competitive. Indeed, the result naturally extends to the more general setting where the algorithm has zero information about the monopolist utilities of the agents. The details can be found in \Cref{app:positive-regime}.
	
	\item For $-1\le p\le \nicefrac{1}{\log n}$, the competitive results (\Cref{thm:nashian}, \Cref{cor:nashian}, and \Cref{thm:nashian-to-harmonic}) are preserved with a slight adjustment to Nashian Greedy.
	
	\item For $-\infty\le p\le -1$, the competitive result (\Cref{thm:harmonic-to-egalitarian}) is preserved with a slight adjustment to Mixed Greedy.
\end{enumerate}

In the remainder of this section, we discuss how to maintain the competitive results in this general setting for $-\infty \leq p \leq \nicefrac{1}{\log n}$.

\subsection{Preliminaries}

We first restate the definition of Uniform Allocation.

\begin{algorithm}{Uniform Allocation}
	For each item $t \in I$, allocate it uniformly to all agents, i.e., $x_a(t) = \nicefrac{1}{n}$.
\end{algorithm}

The base utility guaranteed by Uniform Allocation now varies for each individual agent based on their own monopolistic utilities.

\begin{lemma}[c.f.\ \Cref{lem:unif}]
	\label{lem:unif-app}
	Every agent $a \in A$ gets utility $\nicefrac{V_a}{n}$ from Uniform Allocation.
\end{lemma}

As before, we consider a relaxation of the online $p$-mean welfare maximization, in which each agent $a$ starts with the base utility $U_a(0) = \nicefrac{V_a}{n}$ for the algorithm's allocation, and the benchmark is the original optimal $p$-mean welfare without the $\nicefrac{V_a}{n}$ base utility.
\begin{lemma}[c.f.\ \Cref{lem:relaxation}]
	\label{lem:relaxation-app}
	If an online algorithm $A$ is $\Gamma$-competitive for the relaxed online $p$-mean welfare maximization problem with accurate predictions of monopolist utilities, then allocating half of each item by Uniform Allocation and the other half by algorithm $A$ is $2 \Gamma$-competitive for the original problem with accurate predictions of monopolist utilities.
\end{lemma}

To apply the reduction in \Cref{lem:relaxation-app}, the initialization of the base utilities of the agents in Nashian Greedy and Mixed Greedy should be modified in accordance with \Cref{lem:unif-app}.

In alignment with the earlier organization, we first re-define Nashian Greedy and defer the re-definition of Mixed Greedy to \Cref{sec:harmonic-to-egalitarian-app}.

\begin{algorithm}{Nashian Greedy (Extended Version)}
	\emph{Initialization:~}
	Let $x_a(t) = 0$, $\utility_a(0) = \frac{V_a}{n}$ for any agent $a \in A$ and any time $t \in I$.\\[2ex]
	\emph{Online Decisions:~}
	Allocate each item $t \in I$ to an agent $a$ that maximizes:
	\[
	\frac{v_a(t)}{\utility_a(t)}
	~,
	\]
	to greedily maximize the (logarithm of) Nash welfare:
	\[
	\sum_{a \in A} \log \utility_a(t)
	~.
	\]
\end{algorithm}

A key observation is that the Fundamental Lemma of Nashian Allocation (\Cref{lem:nashian}) remains valid under this extended definition of the Nashian Greedy algorithm.

\begin{lemma}[Fundamental Lemma of Nashian Allocation, c.f.\ \Cref{lem:nashian}]
	\label{lem:nashian-app}
	Consider any allocation of the items $\tilde{x} = (\tilde{x}_{ai})_{a \in A, i \in I}$.
	For any time $t \in I$ and the agents' utilities for allocation $\tilde{x}$ up to time $t$, denoted as $\tilde{\utility}(t) = (\tilde{\utility}_a(t))_{a \in A}$, we have:
	\[
	\frac{1}{n} \sum_{a \in A} \frac{\tilde{\utility}_a(t)}{\utility_a(t)} \le \log (n+1)
	~.
	\]
\end{lemma}

\begin{proof}
	Consider any item $s \in [0, t)$.
	Suppose that Nashian Greedy allocates it to agent $a$.
	By the definition of the algorithm, for any agent $\tilde{a}$ we have:
	\[
	\frac{v_a(s)}{\utility_a(s)} \ge \frac{v_{\tilde{a}}(s)}{\utility_{\tilde{a}}(s)}
	~.
	\]
	
	The left-hand-side equals the increment of the logarithm of Nash  welfare for Nashian Greedy's allocation.
	Further, we relax the denominator of the right-hand-side to $U_{\tilde{a}}(t)$.
	We get that:
	\[
	\frac{\dif}{\dif{s}} \sum_{a \in A} \log \utility_a(s) \ge \frac{v_{\tilde{a}}(s)}{\utility_{\tilde{a}}(t)}
	~.
	\]
	
	Multiplying this inequality by $\tilde{x}_a(s)$ and summing over all agents $\tilde{a} \in A$, we get that:
	\[
	\frac{\dif}{\dif{s}} \sum_{a \in A} \log \utility_a(s)  \ge \frac{\dif}{\dif{s}} \sum_{a\in A} \frac{\tilde{\utility}_a(s)}{\utility_a(t)}
	\]
	
	Integrating over $s \in [0, t)$ gives:
	\[
	\sum_{a \in A} \log \frac{\utility_a(t)}{\utility_a(0)} \ge \sum_{a \in A} \frac{\tilde{\utility}_a(t)}{\utility_a(t)} 
	~.
	\]
	
	The lemma now follows by $V_a(1+\nicefrac{1}{n}) \ge \utility_a(t) \ge \utility_a(0) = \nicefrac{V_a}{n}$.
\end{proof}

\subsection{Nashian Regime: $-o(1)\leq p\leq \nicefrac{1}{\log n}$}

The same results for $-o(1)\leq p\leq \nicefrac{1}{\log n}$ are established using the extended version of Nashian Greedy.

\begin{theorem}[c.f.\ \Cref{thm:nashian}]
	\label{thm:nashian-app}
	For any $|p| = o(1)$, Nashian Greedy is $O \big( n^{|p|} \cdot \log n \big)$-competitive for the relaxed online $p$-mean welfare maximization problem with accurate predictions of monopolist utilities.
\end{theorem}

\begin{corollary}[c.f.\ \Cref{cor:nashian}]
	\label{cor:nashian-app}	
	For any $p$ such that $|p| \le \nicefrac{1}{\log n}$, Nashian Greedy is $O(\log n)$-competitive for the relaxed online $p$-mean welfare maximization problem with accurate predictions of monopolist utilities.
\end{corollary}

\begin{proof}[Proof of \Cref{thm:nashian-app}]
	For $p = 0$, the theorem follows by considering \Cref{lem:nashian-app} with $\tilde{x} = x^*$, i.e., the optimal allocation, and applying the AM-GM inequality to the left-hand-side.
	
	Next, we consider the case when $p \ne 0$.
	Define an auxiliary allocation $\tilde{x}$ that distributes half of each item uniformly to all agents, and the other half following the optimal allocation $x^*$.
	That is, for any agent $a \in A$ and any item $t \in I$:
	\[
	\tilde{x}_a(t) \defeq \frac{1}{2n} + \frac{1}{2} x^*_a(t)
	~.
	\]
	
	Let $\tilde{U}_a$ be agent $a$'s utility for allocation $\tilde{x}$.
	By definition, we have:
	\begin{equation}
		\label{eqn:auxiliary-utility-range-app}
		\frac{V_a}{2n} \le \tilde{U}_a \le V_a \left(\frac{1}{2} + \frac{1}{2n} \right)
		~,
	\end{equation}
	and also:
	\begin{equation}
		\label{eqn:auxiliary-allocation-approximation-app}
		\bigg( \frac{1}{n} \sum_{a \in A} \tilde{\utility}_a^p \bigg)^{\frac{1}{p}} \ge \frac{\OPT}{2}
		~.
	\end{equation}
	
	We write the $p$-mean welfare of Nashian Greedy's allocation as:
	\begin{equation}
		\label{eq:nashian-p-mean-rewrite-app}	
		\ALG = \bigg( \frac{1}{n} \sum_{a \in A} \utility_a^p \bigg)^{\frac{1}{p}} = \bigg(\frac{1}{n} \sum_{a \in A} \tilde{\utility}_a^p \cdot \Big( \frac{\tilde{\utility}_a}{\utility_a} \Big)^{-p} \bigg)^{\frac{1}{p}}
		~.
	\end{equation}
	
	Next, we introduce a set of auxiliary variables to denote:
	\[
	z_a \defeq \frac{\tilde{\utility}_a^p}{\sum_{a' \in A} \tilde{\utility}_{a'}^p}
	~.
	\]
	
	Comparing \Cref{eqn:auxiliary-allocation-approximation-app,eq:nashian-p-mean-rewrite-app}, it suffices to show that:
	\[
	\bigg( \sum_{a \in A} z_a \cdot \Big( \frac{\tilde{\utility}_a}{\utility_a} \Big)^{-p} \bigg)^{-\frac{1}{p}} \le \big(K(n+1)\big)^{|p|} \cdot \log(n+1)
	~.
	\]
	
	By the definition of these auxiliary variables, we have $\sum_{a \in A} z_a = 1$.
	Hence, by the relation of $(-p)$-mean and $1$-mean, where recall that $-p < 1$, we have:
	\[
	\bigg( \sum_{a \in A} z_a \cdot \Big( \frac{\tilde{\utility}_a}{\utility_a} \Big)^{-p} \bigg)^{-\frac{1}{p}}\le \sum_{a \in A}  z_a \cdot \frac{\tilde{\utility}_a}{\utility_a}
	~.
	\]
	
	Further, by the range of $\tilde{\utility}_a$ in \Cref{eqn:auxiliary-utility-range-app}, we have:
	\[
	z_a \le \frac{1}{n} \bigg( \frac{\max_{a \in A} \tilde{\utility}_a}{\min_{a \in A} \tilde{\utility}_a} \bigg)^{|p|} \le \frac{1}{n} \cdot \big(K(n+1)\big)^{|p|}
	~.
	\]
	
	Combining the above with \Cref{lem:nashian-app}, we get that:
	\[
	\bigg( \sum_{a \in A} z_a \cdot \Big( \frac{\tilde{\utility}_a}{\utility_a} \Big)^{-p} \bigg)^{-\frac{1}{p}}\le \big(K(n+1)\big)^{|p|} \cdot \frac{1}{n} \sum_{a \in A} \frac{\tilde{\utility}_a}{\utility_a} \le \big(K(n+1)\big)^{|p|} \cdot \log(n+1)
	~.
	\]
\end{proof}

\subsection[From Nash to Harmonic Welfare]{From Nash to Harmonic Welfare: $-1 \le p \le - \Omega(1)$}

The same result for $-1\leq p\leq -\Omega(1)$ are also established using the extended version of Nashain Greedy.

\begin{theorem}[c.f.\ \Cref{thm:nashian-to-harmonic}]
	\label{thm:nashian-to-harmonic-app}
	For any $-1 \le p \le -	\Omega(1)$, Nashian Greedy is:
	\[
	O \Big( n^{\frac{|p|}{|p|+1}} (\log n)^{\frac{1}{|p|+1}} \Big)
	\]
	competitive for the relaxed online $p$-mean welfare maximization problem with accurate predictions of monopolist utilities.
\end{theorem}

\subsubsection{Further Properties of Nashian Allocation}

\begin{definition}[Bad Agents, \Cref{def:bad} restated]
	For any time $t\in I$ and any $\beta > 0$, let $B_\beta(t)$ be the set of \emph{$\beta$-bad} agents whose regularized utilities at time $t$ are at most a $\beta$ fraction of the optimal $p$-mean welfare, i.e.:
	\[
	B_{\beta}(t) \defeq \Big\{a\in A:\,\utility_a(t) \leq \beta \cdot \OPT \Big\}
	~.
	\]
	Further, we write $B_\beta$ for $B_\beta(T)$, the set of $\beta$-bad agents at the end.
\end{definition}

\begin{lemma}[c.f.\ \Cref{lem:bad-agents-optimal-utility}]
	\label{lem:bad-agents-optimal-utility-app}
	For any time $t\in I$, any $\beta > 0$, and any subset of $\beta$-bad agents $S \subseteq B_\beta(t)$, we have:
	\[
	\frac{1}{n} \sum_{a\in S} \utility_a^* 
	\leq \beta \log (n+1) \cdot \OPT + \frac{1}{n} \sum_{a\in S} \Big(V_a-\int_0^t v_a(s) \dif{s} \Big)
	~.
	\]
\end{lemma}

\begin{proof}
	Recall that $x^*$ denotes the optimal allocation.
	The left-hand-side of the lemma's inequality can be written as:
	\[
	\frac{1}{n} \sum_{a \in S} \int_0^T v_a(s) x^*_a(s) \dif{s}
	~.
	\]
	
	On one hand, observe that:
	\[
	\int_t^T v_a(s) x^*_a(s) \dif{s} \le \int_t^T v_a(s) \dif{s} = V_a - \int_0^t v_a(s) \dif{s}
	~.
	\]
	
	On the other hand, by \Cref{lem:nashian-app}, we have:
	\[
	\frac{1}{n} \sum_{a \in A} \frac{1}{\utility_a(t)} \int_0^t v_a(s) x^*_a(s) \dif{s} \le \log(n+1) 
	~.
	\]
	
	We now drop all agents outside subset $S$ from the summation on the left-hand-side, and relax $\utility_a(t)$ to its upper bound $\beta \cdot \OPT$ for the remaining $\beta$-bad agents $a \in S$.
	We get that:
	\[
	\frac{1}{n} \sum_{a \in S} \int_0^t v_a(s) x^*_a(s) \dif{s} \le \beta \log (n+1) \cdot \OPT
	~.
	\]	
	
	Putting together these two parts proves the lemma.
\end{proof}

\begin{lemma}[c.f.\ \Cref{lem:bad-agents-number}]
	\label{lem:bad-agents-number-app}
	For any $\beta > 0$, the fraction of $\beta$-bad agents at the end is at most:
	\[
	\frac{|B_\beta|}{n} \le \big( \beta \log (n+1) \big)^{\frac{|p|}{|p|+1}}
	~.
	\]
\end{lemma}

\begin{proof}
	By \Cref{lem:bad-agents-optimal-utility-app} with $t = T$ and $S = B_\beta$, we have:
	\[
	\frac{1}{n} \sum_{a \in B_\beta} \utility_a^* \leq \beta \log (n+1) \cdot \OPT
	~.
	\]	
	
	Further, recall that $p < 0$.
	We have:
	\[
	\OPT^p = \frac{1}{n} \sum_{a \in A} (\utility_a^*)^{-|p|} \ge \frac{1}{n} \sum_{a \in B_\beta} (\utility_a^*)^{-|p|} 
	~.
	\]
	
	Finally, by H\"{o}lder's inequality:
	\[
	\bigg( \frac{1}{n} \sum_{a \in B_\beta} \utility_a^* \bigg)^{\frac{|p|}{|p|+1}} 
	\bigg( \frac{1}{n} \sum_{a \in B_\beta} (\utility_a^*)^{-|p|} \bigg)^{\frac{1}{|p|+1}} \ge \frac{1}{n} \sum_{a \in B_\beta} 1 = \frac{|B_\beta|}{n}
	~.
	\]
	
	Combining these inequalities proves the lemma.
\end{proof}

\subsubsection{Proof of \Cref{thm:nashian-to-harmonic-app}}

We compare the $p$-mean welfare of Nashian Greedy to the optimal benchmark by:
\[
\bigg( \frac{\ALG}{\OPT} \bigg)^p = \frac{1}{n} \sum_{a\in A} \bigg( \frac{\utility_a}{\OPT} \bigg)^p
~.
\]

By $\utility_a \ge \nicefrac{V_a}{n} \ge \nicefrac{\min_{a\in A} V_a}{n}$ and $\OPT \le \max_{a\in A} V_a$, and recalling that $p < 0$, we have :
\[
\bigg(\frac{\utility_a}{\OPT}\bigg)^p \le (Kn)^{|p|}
~.
\]

Therefore, the above ratio can be written as:
\begin{align*}
	\bigg( \frac{\ALG}{\OPT} \bigg)^p
	& = \int_0^{(Kn)^{|p|}} \big(\text{fraction of agents with $\big( \nicefrac{\utility_a}{\OPT} \big)^p \ge \alpha$}\big) \dif{\alpha} \\
	&
	\le \int_0^{(Kn)^{|p|}} \big( \alpha^{\frac{1}{p}} \log(n+1) \big)^{\frac{|p|}{|p|+1}} \dif{\alpha} 
	\tag{\Cref{lem:bad-agents-number-app}} \\
	&
	= \int_0^{(Kn)^{|p|}} \big( \log(n+1) \big)^{\frac{|p|}{|p|+1}} \alpha^{-\frac{1}{|p|+1}} \dif{\alpha} \\
	&
	= \frac{|p|+1}{|p|} \cdot (Kn)^{\frac{|p|^2}{|p|+1}} \big(\log(n+1) \big)^{\frac{|p|}{|p|+1}}
	~.
\end{align*}

Taking $p$-th root on both sides gives:
\[
\frac{\ALG}{\OPT} ~\ge \underbrace{\bigg( \frac{|p|}{|p|+1} \bigg)^{\frac{1}{|p|}}}_{\text{$\Omega(1)$ for $|p| = \Omega(1)$}} \cdot~ (Kn)^{-\frac{|p|}{|p|+1}} \big(\log(n+1) \big)^{-\frac{1}{|p|+1}}
~.
\]

\subsection[From Harmonic to Egalitarian Welfare]{From Harmonic to Egalitarian Welfare: $-\infty \le p \le -1$}
\label{sec:harmonic-to-egalitarian-app}

We now present the extended definition of this Mixed Greedy algorithm that combines Nashian Greedy and Regularized Egalitarian Greedy.
In line with the earlier definition, we further relax the problem by letting there be two copies of each item.
This relaxation only affects the competitive ratio by a constant factor.

\begin{algorithm}{Mixed Greedy (Extended Version)}
	\emph{Initialization:~}
	Let $x_a(t) = 0$, $\utility_a(0) = \frac{V_a}{n}$ for any agent $a \in A$ and any time $t \in I$.\\[2ex]
	\emph{Regularization:~}
	Define the regularizer of any agent $a \in A$ as:
	\[
	\regularizer_a(t) = \frac{1}{\Phi} 
	\Big(V_a - \int_0^t v_a(s) \dif{s} \Big)
	~,
	\]
	where $\Phi = \sqrt{n \log(n+1)}$.\\[2ex]
	\emph{Online Decisions:~}
	For each item $t \in I$, which comes in two copies:
	\begin{enumerate}
		\item Allocate the \emph{Nashian copy} using Nashian Greedy.
		\item Allocate the \emph{egalitarian copy} to the agents with the smallest regularized utility:
		\[
		\utility_a(t) + \regularizer_a(t)
		~,
		\]
		to greedily maximize the regularized egalitarian welfare:
		\[
		\min_{a \in A} \big( \utility_a(t) + \regularizer_a(t) \big)
		~.
		\]
	\end{enumerate}
\end{algorithm}

The same result for $-\infty \le p\le -1$ is established using the extended version of Mixed Greedy.

\begin{theorem}[c.f.\ \Cref{thm:harmonic-to-egalitarian}]
	\label{thm:harmonic-to-egalitarian-app}
	For any $-\infty \le p \le -1$, Mixed Greedy is $O \big( \sqrt{n \log n} \big)$-competitive for the relaxed online $p$-mean welfare maximization problem.
\end{theorem}

\subsubsection{Properties of the Regularized Utilities}

\begin{definition}[Critical Agents, \Cref{def:critical} restated]
	For any time $t \in I$, and for any $\beta  \geq 0$, let $C_\beta(t)$ be the set of \emph{$\beta$-critical} agents at time $t$, whose regularized utilities are at most a $\beta$ fraction of the optimal $p$-mean welfare, i.e.:
	\[
	C_{\beta}(t) = \Big\{a\in A:\,\utility_a(t)+\regularizer_a(t) \leq \beta \cdot \OPT \Big\}
	~.
	\]
\end{definition}

\begin{lemma}[c.f.\ \Cref{lem:critical-agent-number}]
	\label{lem:critical-agent-number-app}
	For any time $t\in I$ and any $\beta > 0$, the fraction of $\beta$-critical agents at time $t$ is upper bounded by:
	\[
	\frac{|C_\beta(t)|}{n} \le \max \Big\{ \, \big(2 \beta \log(n+1) \big)^{\frac{|p|}{|p|+1}} \,,\, (2\Phi\beta)^{|p|} \,\Big\}
	~.
	\]
\end{lemma}

\begin{proof}    
	On one hand, by the definition of $p$-mean welfare and $p < 0$, we have:
	\[
	\OPT^p = \frac{1}{n} \sum_{a\in A} \big(\utility_a^*\big)^{-|p|} \ge \frac{1}{n} \sum_{a \in C_\beta(t)} \big(\utility_a^*\big)^{-|p|}
	~.
	\]
	
	On the other hand, all $\beta$-critical agents are $\beta$-bad.
	We have:
	\begin{align*}
		\frac{1}{n} \sum_{a \in C_\beta(t)} \utility_a^*
		&
		\leq \beta \log (n+1) \cdot \OPT + \frac{1}{n} \sum_{a \in C_\beta(t)} \Big( V_a - \int_0^t v_a(s) \dif{s} \Big)
		\tag{\Cref{lem:bad-agents-optimal-utility-app}} \\
		&
		= \beta \log (n+1) \cdot \OPT + \frac{1}{n} \sum_{a \in C_\beta(t)} \Phi \cdot \regularizer_a(t) \\
		&
		\le \bigg(\beta \log (n+1) + \frac{|C_\beta(t)|}{n} \cdot \Phi \beta \bigg) \cdot \OPT
		~.
		\tag{$\regularizer_a(t) \le \utility_a(t) + \regularizer_a(t) \le \beta \cdot \OPT$}
	\end{align*}
	
	By H\"{o}lder's inequality:
	\[
	\bigg( \frac{1}{n} \sum_{a \in C_\beta(t)} \big(\utility_a^*\big)^{-|p|} \bigg)^{\frac{1}{|p|+1}} \bigg( \frac{1}{n} \sum_{a \in C_\beta(t)} \utility_a^* \bigg)^{\frac{|p|}{|p|+1}} \ge \frac{1}{n} \sum_{a \in C_\beta(t)} 1 = \frac{|C_\beta(t)|}{n}
	~.
	\]
	
	Combining these inequalities gives:
	\[
	\bigg( \beta \log (n+1) + \frac{|C_\beta(t)|}{n} \cdot \Phi \beta \bigg)^{\frac{|p|}{|p|+1}} \ge \frac{|C_\beta(t)|}{n}
	~.
	\]
	
	Rearranging terms, we have:
	\[
	\beta \log (n+1) \cdot \bigg( \frac{n}{|C_\beta(t)|} \bigg)^{\frac{|p|+1}{|p|}} + \Phi\, \beta \cdot \bigg( \frac{n}{|C_\beta(t)|} \bigg)^{\frac{1}{|p|}} \ge 1
	~.
	\]
	
	Hence, either the first part is at least a half, in which case:
	\[
	\frac{|C_\beta(t)|}{n} \le \big(2 \beta \log(n+1) \big)^{\frac{|p|}{|p|+1}}
	~,
	\]
	or the second part is at least a half, in which case:
	\[
	\frac{|C_\beta(t)|}{n} \le (2\Phi\beta)^{|p|}
	~.
	\]
	
	In either case, the lemma follows.
\end{proof}

\begin{corollary}[c.f.\ \Cref{cor:critical-threshold}]
	\label{cor:critical-threshold-app}
	For any time $t \in I$, and:
	\begin{equation}
		\label{eq:beta-star-app}
		\beta^* = \frac{1}{2} \cdot  n^{-\frac{1}{2}-\frac{1}{2|p|}} (\log (n+1))^{-\frac{1}{2}+\frac{1}{2|p|}}
		~,	
	\end{equation}
	we have:
	\[
	\big|C_{\beta^*} (t)\big| \le \sqrt{n \log (n+1)}
	~.
	\]
\end{corollary}

Using \Cref{cor:critical-threshold-app}, we derive a universal lower bound for all agents' utilities.

\begin{lemma}[c.f.\ \Cref{lem:negative-infinity-main-lemma}]
	\label{lem:negative-infinity-main-lemma-app}
	For the choice of $\beta^*$ in \Cref{eq:beta-star-app}, and any agent $a \in A$, we have:
	\[
	\utility_a \ge \frac{\beta^*}{K} \cdot \OPT
	~.
	\]
\end{lemma}

\begin{proof}
	We will prove a stronger claim that for any agent $a \in A$ and any time $t \in I$:
	\begin{equation}
		\label{eqn:regularized-utility-app}
		\utility_a(t) + \regularizer_a(t) \ge \frac{\beta^*}{K} \cdot \OPT
		~.
	\end{equation}
	
	Then, the lemma holds as the special case when $t = T$ because $\regularizer_a(T) = 0$.
	
	Initially at time $t = 0$, we have:
	\[
	\regularizer_a(0) = \frac{V_a}{\Phi} \ge \frac{1}{\Phi} \min_{a\in A} V_a > \frac{\beta^*}{K} \cdot \max_{a\in A} V_a \ge \frac{\beta^*}{K} \cdot \OPT
	~.
	\]
	
	To prove that \Cref{eqn:regularized-utility-app} holds at all time $t$, it suffices to show that for any time $t \in I$ when there is at least one $\beta^*$-critical agent $a \in C_{\beta^*}(t)$, the allocation of the egalitarian copy of item $t$ weakly increases the regularized egalitarian welfare.
	
	For any critical agent $a \in C_{\beta^*}(t)$, we have:
	\[
	\frac{\dif}{\dif{t}} \regularizer_a (t) = - \frac{v_a(t)}{\Phi}
	~.
	\] 
	
	Further, \Cref{cor:critical-threshold-app} asserts that at most $\sqrt{n \log(n+1)}$ agents are $\beta^*$-critical.
	Hence, allocating the egalitarian copy of item $t$ equally among these $\beta^*$-critical agents would have yielded:
	\[
	\frac{\dif}{\dif{t}} \utility_a(t) \ge \frac{v_a(t)}{\sqrt{n \log (n+1)}} = \frac{v_a(t)}{\Phi}
	~,
	\]
	and weakly increased the regularized egalitarian welfare. 
	The greedy allocation of the algorithm would only do better.
\end{proof}

\subsubsection{Proof of \Cref{thm:harmonic-to-egalitarian-app}}

The proof is almost verbatim to that of \Cref{thm:nashian-to-harmonic-app}, except that we will use the newly developed \Cref{lem:negative-infinity-main-lemma-app} to lower bound the agents' utilities, replacing the basic bound $\utility_a \ge \nicefrac{1}{n}$.

Consider the $p$-th power of the Mixed Greedy algorithm's $p$-mean welfare, normalized by the $p$-th power of $\OPT$:
\[
\bigg( \frac{\ALG}{\OPT} \bigg)^p = \frac{1}{n} \sum_{a\in A} \bigg( \frac{\utility_a}{\OPT} \bigg)^p
~.
\]

By \Cref{lem:negative-infinity-main-lemma-app} and $p < 0$, we have:
\[
\bigg(\frac{\utility_a}{\OPT}\bigg)^p \le \bigg(\frac{\beta^*}{K}\bigg)^p
~.
\]

The above ratio can therefore be written as:
\begin{align*}
	\bigg( \frac{\ALG}{\OPT} \bigg)^p
	& = \int_0^{(\nicefrac{\beta^*}{K})^p} (\text{fraction of agents with $\big( \nicefrac{\utility_a}{\OPT} \big)^p \ge \alpha$)} \dif{\alpha} \\
	&
	\le \int_0^{(\nicefrac{\beta^*}{K})^p} \big( \alpha^{\frac{1}{p}} \log(n+1) \big)^{\frac{|p|}{|p|+1}} \dif{\alpha} 
	\tag{\Cref{lem:bad-agents-number-app}} \\
	&
	= \int_0^{(\nicefrac{\beta^*}{K})^p} \big( \log(n+1) \big)^{\frac{|p|}{|p|+1}} \alpha^{-\frac{1}{|p|+1}} \dif{\alpha} \\
	&
	= \frac{|p|+1}{|p|} \bigg(\frac{\beta^*}{K}\bigg)^{-\frac{|p|^2}{|p|+1}} \big(\log(n+1) \big)^{\frac{|p|}{|p|+1}} \\
	&
	= \frac{|p|+1}{|p|} (2K)^{{\frac{|p|^2}{|p|+1}}} \big( n \log(n+1) \big)^{\frac{|p|}{2}}
	~.
\end{align*}

Taking $p$-th root on both sides gives:
\[
\frac{\ALG}{\OPT} ~\ge~ \underbrace{\bigg( \frac{|p|}{|p|+1} \bigg)^{\frac{1}{|p|}} (2K)^{-{\frac{|p|}{|p|+1}}}}_{\text{$\Omega(1)$ for $p \le -1$}} ~\cdot~ {\Big(\sqrt{n \log(n+1)}\Big)}^{-1}
~.
\]

\end{document}